\documentclass{sig-alt-full}
\usepackage{url}
\usepackage{epsf}
\usepackage{graphicx}
\usepackage{amsfonts}
\usepackage{amssymb}
\usepackage{latexsym}
\usepackage{setspace}
\newtheorem{theorem}{{\bf Theorem}}

\newtheorem{lemma}[theorem]{Lemma}

\newcommand{\R}{\mathcal R}

\newcommand{\ssp}{{\tt sp}}

\title{Local Approximation Schemes for Topology Control}
\numberofauthors{2}
\author{
\alignauthor Mirela Damian \\
\affaddr{Dept. of Comp. Sci, Villanova Univ.}\\
\affaddr{Villanova, PA 19085} \email{mirela.damian@villanova.edu}
\alignauthor Saurav Pandit
~~~Sriram Pemmaraju\\
\affaddr{Dept. of Comp. Sci, Univ. of Iowa}\\
\affaddr{Iowa City, IA 52242-1419} \\
\email{[spandit, sriram]@cs.uiowa.edu} }

\begin{document}

%\conferenceinfo{PODC'06,} {July 22-26, 2006, Denver, Colorado, USA.}

\conferenceinfo{Altered version of PODC'06, with a few typos fixed.}
{July 22-26, 2006, Denver, Colorado, USA.}

\CopyrightYear{2006}

\crdata{1-59593-384-0/06/0007}

\maketitle

\begin{abstract}
This paper presents a distributed algorithm on wireless ad-hoc
networks that runs in polylogarithmic number of rounds in the size
of the network and constructs a linear size, lightweight,
$(1+\varepsilon)$-spanner for any given $\varepsilon > 0$. A
wireless network is modeled by a $d$-dimensional $\alpha$-quasi unit
ball graph ($\alpha$-UBG), which is a higher dimensional
generalization of the standard unit disk graph (UDG) model. The
$d$-dimensional $\alpha$-UBG model goes beyond the unrealistic
``flat world'' assumption of UDGs and also takes into account
transmission errors, fading signal strength, and physical
obstructions. The main result in the paper is this: for any fixed
$\varepsilon > 0$, $0 < \alpha \le 1$, and $d \ge 2$ there is a
distributed algorithm running in $O(\log n \cdot \log^* n)$
communication rounds on an $n$-node, $d$-dimensional $\alpha$-UBG
$G$ that computes a $(1+\varepsilon)$-spanner $G'$ of $G$ with
maximum degree $\Delta(G') = O(1)$ and total weight $w(G') =
O(w(MST(G))$. This result is motivated by the topology control
problem in wireless ad-hoc networks and improves on existing
topology control algorithms along several dimensions. The technical
contributions of the paper include a new, sequential, greedy
algorithm with relaxed edge ordering and lazy updating, and
clustering techniques for filtering out unnecessary edges.
\end{abstract}

\smallskip
\noindent {\bf Categories and Subject Descriptors:}
%\\
%\noindent
%C. COMPUTER SYSTEMS ORGANIZATION
%\\
%\noindent
%C.2. COMPUTER-COMMUNICATION NETWORKS
C.2.4 [Computer-Communication Networks]: Distributed Systems

\noindent
{\bf General Terms:} Algorithms, Performance, Theory.


\section{Introduction}

Let $G = (V, E)$ be a graph with edge weights $w : E \to \R^+$. For
$t \ge 1$, a $t$-spanner of $G$ is a spanning subgraph $G'$ of $G$
such that for all pairs of vertices $u, v \in V$, the length of a
shortest $uv$-path in $G'$ is at most $t$ times the length of a
shortest $uv$-path in $G$. The problem of constructing a sparse
$t$-spanner, for small $t$, of a given graph $G$ has been
extensively studied by researchers in distributed computing and
computational geometry and more recently by researchers in ad-hoc
wireless networks. In this paper we present a fast distributed
algorithm for constructing a linear size, lightweight $t$-spanner of
bounded degree for any given $t > 1$, on wireless networks. Below,
we describe our result more precisely.

\subsection{Network model}
We model wireless networks using $d$-dimensional quasi unit ball
graphs. For any fixed $\alpha$, $0 < \alpha \le 1$ and integer $d
\ge 2$, a {\em $d$-dimensional $\alpha$-quasi unit ball graph\/}
({\em $\alpha$-UBG}, in short) is a graph $G = (V, E)$ whose vertex
set $V$ can be placed in one-one correspondence with a set of points
in the $d$-dimensional Euclidean space and whose edge set $E$
satisfies the constraint: if $|uv| \le \alpha$ then $\{u, v\} \in E$
and if $|uv| > 1$ then $\{u, v\} \not\in E$. Here we use $|uv|$ to
denote the Euclidean distance between the points corresponding to
vertices $u$ and $v$. The $\alpha$-UBG model does not prescribe
whether a pair of vertices whose distance is in the range $(\alpha,
1]$ are to be connected by an edge or not. This is an attempt to
take into account transmission errors, fading signal strength, and
physical obstructions. Our algorithm does not need to know the
locations of nodes of the $\alpha$-UBG in $d$-dimensional Euclidean
space; just the pairwise Euclidean distances.

The $\alpha$-UBG model is a higher dimensional generalization of the
somewhat simplistic unit disk graph (UDG) model of wireless networks
that is popular in literature. Specifically, when $\alpha = 1$ and
$d = 2$, a $d$-dimensional $\alpha$-UBG is just a UDG. UDGs are
attractive due to their mathematical simplicity, but have been
deservedly criticized for being unrealistic models of wireless
networks \cite{KotzNewportElliot}. In our view, $d$-dimensional
$\alpha$-UBGs are a significant step towards a more realistic model
of wireless networks. Two-dimensional $\alpha$-UBGs were proposed in
\cite{BarriereFraigniaudNarayanan} as a model of wireless ad-hoc
networks with unstable transmission ranges and the difficulty of
doing geometric routing in such networks was shown.

%\subsection{Communication Model}
Our communication model is the standard synchronous message passing
model that does not account for channel access and collision issues.
In this communication model, time is divided into rounds. In each
round, each node can send a different message to each of its
neighbors, receive different messages from all neighbors and perform
arbitrary (polynomial) local computation. The length of messages
exchanged between nodes is logarithmic in the number of nodes. We
measure the cost of our algorithm in terms of the number of
communication rounds. Although this model is not widely considered
to be realistic, it is nevertheless interesting because it
demonstrates the locality of computations.

\newpage
\subsection{Our result}
For any edge weighted graph $J$, we use $w(J)$ to denote the sum of
the weights of all the edges in $J$ and $MST(J)$ to denote a minimum
weight spanning tree of $J$. For any fixed $\varepsilon > 0$, $0 <
\alpha \le 1$, and $d \ge 2$ our algorithm runs in $O(\log n \cdot
\log^* n)$ communication rounds on an $n$-node, $d$-dimensional
$\alpha$-UBG and computes a $(1+\varepsilon)$-spanner $G'$ of $G$
whose maximum degree $\Delta(G') = O(1)$ and whose total weight
$w(G') = O(w(MST(G))$. Since any spanner of $G$ has weight bounded
below by $w(MST(G))$, the weight of the output produced by the
algorithm is within a constant times the optimal weight. As far as
we know, our result significantly improves all known results of a
similar kind along several dimensions. More on this further below.

\subsection{Topology control}
Our result is motivated by the {\em topology control\/} problem in
wireless ad-hoc networks. For an overview of topology control, see
the survey by Rajaraman \cite{Rajaraman}. Since an ad-hoc network
does not come with fixed infrastructure, there is no topology to
start with and informally speaking, the topology control problem is
one of selecting neighbors for each node so that the resulting
topology has a number of useful properties. More precisely, let $V$
be a set of nodes that can communicate via wireless radios and for
each $v \in V$, let $N(v)$ denote the set of all nodes that $v$ can
reach when transmitting at maximum power. The induced digraph $G =
(V, E)$, where $E = \{\{u, v\} \mid v \in N(u)\}$, represents the
network in which every node has chosen to transmit at maximum power
and has designated every node it can reach as its neighbor. The
topology control problem is the problem of devising an efficient and
local protocol $P$ for selecting a set of neighbors $N_P(v)
\subseteq N(v)$ for each node $v \in V$. The induced digraph $G_P =
(V, E_P)$, where $E_P = \{\{u, v\} \mid v \in N_P(u)\}$ is typically
required the satisfy properties such as symmetry (if $v \in N_P(u)$
then $u \in N_P(v)$), sparseness ($|E_P| = O(|V|)$) or bounded
degree ($|N_P(v)| \le c$ for all nodes $v$ and some constant $c$),
and the spanner property. Sometimes stronger versions of
connectivity such as $k$-vertex connectivity or $k$-edge
connectivity (for $k > 1$) are desired, both for providing
fault-tolerance and for improving throughput
\cite{HajiaghayiImmorlicaMirrokni2002,HajiaghayiImmorlicaMirrokni}.
If the input graph consists of nodes in the plane, it is quite
common to require that the output graph be planar
\cite{LiCalinescuWan,LiCalinescuWanWang,LiWangIJCGA,WangLi,WattenhoferZollinger}.
This requirement is motivated by the existence of simple,
memory-less, geometric routing algorithms that guarantee message
delivery only when the underlying graph is planar \cite{KarpKung}.

Though the topology control problem is recent, there is already an
extensive body of literature on the problem to which the above sample
of citations do not do justice.
%\cite{WattenhoferZollinger,LiHalpernBahlWangWattenhofer,HajiaghayiImmorlicaMirrokni2002,WangLi,GhoshLillisPanditPemmaraju,SongWangLiFrieder,LiWangIJCGA}.
However, many of the topology control protocols that provide worst
case guarantees on the quality of the topology, assume that the
network is modeled by a UDG. A recent example \cite{LiWangIJCGA}
presents a distributed algorithm that requires a linear number of
communication rounds in the worst case to compute a planar
$t$-spanner of a given UDG with $t \approx 6.2$ and in which each
node has degree at most 25. These two constants can be slightly
tuned -- $t$ can be brought down to about 3.8 with a significant
increase in the degree bound. We improve on the result in
\cite{LiWangIJCGA} along several dimensions. As is generally known
among practitioners in ad-hoc wireless networks, the ``flat world''
assumption and the identical transmission range assumption of UDGs
are unrealistic \cite{KotzNewportElliot}. By using an $\alpha$-UBG
we significantly generalize our model of wireless networks,
hopefully moving much closer to reality. For any $\varepsilon > 0$,
our algorithm returns a $(1 + \varepsilon)$-spanner; as far as we
know, this is the first distributed algorithm that produces an {\em
arbitrarily good\/} spanner for an $\alpha$-UBG model of wireless
networks. We also guarantee that the total weight of the output is
within constant times optimal -- a guarantee that is not provided in
\cite{LiWangIJCGA}. Finally, using algorithmic techniques and
distributed data structures that might be of independent interest,
we ensure that our protocol runs in $O(\log n \cdot \log^* n)$
communication rounds. We are not aware of any topology control
algorithm that runs in poly-logarithmic number of rounds and
provides anywhere close to the guarantees provided by our algorithm.

\subsection{Spanners in computational geometry}
Starting in the early 1990's, researchers in computational geometry
have attempted to find sparse, lightweight spanners for complete
Euclidean graphs. Given a set $P$ of $n$ points in $\R^d$, the tuple
$(P, E)$, where $E$ is the set of line segments $\{\{p, p'\} \mid p,
p' \in P\}$, is called the {\em complete Euclidean graph\/} on $P$.
For any subset $E' \subseteq E$, $(P, E')$ is called a {\em
Euclidean graph\/} on $P$. The specific problem that researchers in
computational geometry have considered, is this. Given a set $P$ of
$n$ points in $\R^d$ and $t > 1$, compute a Euclidean graph on $P$
that is a $t$-spanner of the complete Euclidean graph on $P$, whose
maximum degree is bounded by $O(1)$ and whose weight is bounded by
the weight of a minimum spanning tree on $P$. For an early example,
see \cite{LevcopoulosLingas} in which the authors show that there
are ``planar graphs almost as good as the complete graphs and almost
as cheap as minimum spanning trees.'' This was followed by a series
of improvements
\cite{CzumajZhao,DasHeffernanNarasimhan,DasNarasimhan97,GudmundssonLevcopoulosNarasimhan},
with the most recent paper \cite{CzumajZhao} presenting algorithms
for constructing Euclidean subgraphs that provide the additional
property of $k$-fault tolerance. Most of the papers mentioned above
start with the following simple, greedy algorithm.

%\smallskip\noindent \setbox0\vbox{ \noindent
\begin{center}\begin{minipage}{\linewidth}
\hrule\hfill\\
\noindent Algorithm {\tt SEQ-GREEDY} ($G=(V,E),t$)
\hrule\hfill\\
\noindent
1. Order the edges in $E$ in non-decreasing order of length.\\
2. $E' \leftarrow \phi$, $G' \leftarrow (V,E')$\\
3. For each edge $e = \{u, v\} \in E$ if there is no $uv$-path in $G'$ of length at most $t\cdot |uv|$
\begin{enumerate}
\item[(a)] $E' \leftarrow E'\cup \{ e \}$
\item[(b)] $G' \leftarrow (V,E')$
\end{enumerate}
Output $G'$.
%}\fbox{\box0}
\smallskip
\hrule\hfill
\end{minipage}\end{center}

\noindent It is well-known \cite{DasNarasimhan97} that if the input
graph $G = (V, E)$ is the complete Euclidean graph, then the output
graph $G' = (V, E')$ produced by {\tt SEQ-GREEDY} has the following
useful properties: (i) $G'$ is a $t$-spanner of $G$, (ii)
$\Delta(G') = O(1)$, and (iii) $w(G') = O(w(MST(G)))$. A naive
implementation of {\tt SEQ-GREEDY} takes $O(n^3 \log n)$ time
because a quadratic number of shortest path queries need to be
answered on a dynamic graph with $O(n)$ edges. Consequentially,
papers in this area
\cite{DasNarasimhan97,GudmundssonLevcopoulosNarasimhan} focus on
trying to implement {\tt SEQ-GREEDY} efficiently. For example, Das
and Narasimhan \cite{DasNarasimhan97} show how to use certain kind
of graph clustering to answer shortest path queries efficiently,
thereby reducing the running time of {\tt SEQ-GREEDY} to $O(n \log^2
n)$. One of the contributions of this paper is to show how a variant
of the Das-Narasimhan clustering scheme can be implemented and
maintained efficiently, in a distributed setting.

\subsection{Summary of our contributions}
In obtaining the main result, our paper makes the following
contributions.
\begin{enumerate}
\item We first show that sparse, lightweight $t$-spanners for arbitrarily
small $t > 1$, not only exist for $d$-dimensional $\alpha$-UBGs, but
%such spanners
can be computed using {\tt SEQ-GREEDY}. Note that sparse
$t$-spanners for arbitrarily small values of $t \ge 1$ do not exist
for general graphs. For example, there is a classical
graph-theoretic result that shows that for any $t \ge 1$, there
exist (infinitely many) unweighted $n$-vertex graphs for which every
$t$-spanner needs $\Omega(n^{1 + 1/(t+2)})$ edges (see Page 179 in
\cite{Peleg}).

\item We then consider a version of {\tt SEQ-GREEDY} in which
the requirement that edges be considered in increasing order of
length is relaxed. More precisely, the edges are distributed into
$O(\log n)$ bins $B_0, B_1, B_2, \ldots$ such that edges in $B_i$
are all shorter than edges in $B_{i+1}$. It is then shown that {\em
any\/} ordering of the edges in which edges in $B_0$ come first,
followed by edges in $B_1$, followed by the edges in $B_2$, etc., is
good enough for the correctness of {\tt SEQ-GREEDY}, even for
$d$-dimensional $\alpha$-UBGs. More importantly, we show that the
update step in {\tt SEQ-GREEDY} (Step 3(a)) need not be performed
after each edge is queried. Instead, a more lazy update may be
performed, after each bin is completely processed. Being able to
perform a lazy update is critical for a distributed implementation;
roughly speaking, we want the nodes to query all edges in a bin in
parallel and not to have to rely on answers to queries on other
edges in a bin.

\item We also use a clustering technique as a way to reduce the number
of edges to be queried per node. Reducing the number of query edges
per node, is critical to being able to guarantee that the output of
our distributed version of {\tt SEQ-GREEDY} does not have too many
edges incident on a node.

\item Our next contribution is to show that this relaxed version of {\tt SEQ-GREEDY}
can be implemented in a distributed setting in $O(\log n)$ phases
--- one phase corresponding to each bin --- such that each phase
requires $O(\log^* n)$ communication rounds. Each phase requires the
computation of maximal independent sets (MIS) on some derived
graphs. We show that the derived graphs are unit ball graphs of {\em
constant doubling dimension\/} \cite{KuhnMoscibrodaWattenhofer} and
use the $O(\log^* n)$-round MIS algorithm of Kuhn et al \cite{KuhnMoscibrodaWattenhofer}.
\end{enumerate}

\subsection{Extensions to our main result} \label{subsection:extensions.to.main}

Here we briefly report on extensions to our main result that we have
obtained. They do not appear in this paper due to lack of space.

\begin{enumerate}
\item
Let $G = (V, E)$ be an edge-weighted graph. For any $t > 1$ and
positive integer $k$, a {\em $k$-vertex fault-tolerant
$t$-spanner\/} of $G$ is a spanning subgraph $G'$ if for each subset
$S$ of vertices of size at most $k$, $G[V' \setminus S]$ is a
$t$-spanner of $G[V \setminus S]$. A $k$-edge fault-tolerant
$t$-spanner is defined in a similar manner. Using ideas from
\cite{CzumajZhao} we can extend our algorithm to produce a
$k$-vertex (or a $k$-edge) fault-tolerant $t$-spanner in
polylogarithmic number of communication rounds.

\item In this paper, we use Euclidean distances as weights for
the edges of the input graph $G$. However, if the metric $c\cdot
|uv|^{\gamma}$, for positive constant $c$ and $\gamma \ge 1$, is
used in place of Euclidean distances $|uv|$, we can show that our
algorithm still produces a spanner with all three desired
properties. Relative Euclidean distances, such as the function
mentioned above, may be used to produce {\em energy spanners}.

\item Let $G = (V, E)$ be an edge-weighted graph.
The {\em power cost\/} of a vertex $u \in V$ %, denoted $power(u)$,
is $power(u) = \max\{ w(u, v) \mid v\mbox{ is a neighbor of }u\}$.
In other words, the power cost of a vertex $u$ is proportional to
the cost of $u$ transmitting to a farthest neighbor. The {\em power
cost\/} of $G$ is $\sum_{u \in V} power(u)$ \cite{HKMN}. We can show
that the output of our algorithm is not only lightweight with
respect to the usual weight measure (sum of the weights of all
edges) but also with respect to the power cost measure.

\end{enumerate}

\section{Sequential Relaxed Greedy Algorithm}
\label{sec:relaxedGreedyAlgorithm}

Now we show that a {\em relaxed version} of {\tt SEQ-GREEDY}
produces an output $G'$ with all three desired properties, even when
the input is not a complete Euclidean graph, but is a
$d$-dimensional, $\alpha$-UBG for fixed $d$ and $\alpha$. Relaxing
the requirement in {\tt SEQ-GREEDY} that the edges be totally
ordered by length and allowing for the output to be updated lazily
are critical to obtaining a distributed algorithm that runs in
polylogarithmic number of rounds.

Let $ r > 1$ be a constant to be fixed later and let $W_i = r^i
\alpha/n$ for each $i = 0, 1, 2, \ldots$. Let $I_0 = (0, \alpha/n]$
and for each $i = 1, 2, \ldots$ let $I_i = (W_{i-1}, W_i]$. Let $m =
\lceil \log_r \frac{n}{\alpha} \rceil$. Then, since no edge has
length greater than 1, the length of any edge in $E$ lies in one of
the intervals $I_0, I_1, \ldots, I_m$. Let $E_i = \{\{u, v\} \in E:
|uv| \in I_i\}$.
%Figure \ref{fig:annulus} illustrates this edge partition.

%\begin{figure}[htpb]
%\epsfxsize=1.6in \centerline{\epsfbox{Figures/annulus.eps}}
%\caption{The figure shows an annulus around $u$ defined by the
%interval $I_i = (W_{i-1}, W_i]$.}
%This annulus is further
%partitioned into $T$ sectors, by cones of angle at most $\theta$,
%that partition the angle around $u$. The set of edges $\{u, v\}$ for
%all vertices $v$ that lie in Sector $i$ is denoted $E_i^i[u]$.}
%\label{fig:annulus}
%\end{figure}

We now eliminate the restriction that edges within a set $E_i$ be
processed in increasing order by length. We run {\tt SEQ-GREEDY} in
$m+1$ phases: in phase $i$, the algorithm processes edges in $E_i$
in arbitrary order and adds a subset of edges in $E_i$ to the
spanner. For $0 \le i \le m$, we use $G_i$ to denote the spanning
subgraph of $G$ consisting of edges $E_0 \cup E_1 \cup \cdots \cup
E_i$. Thus $G_i$ is the portion of the input graph that the
algorithm has processed in phase $i$ and earlier. We use $G_i'$ to
denote the output of the algorithm at the end of phase $i$. In other
words, $G_i'$ is the spanning subgraph of $G$ consisting of edges of
$G$ that the algorithm has decided to retain in phases $0, 1,
\ldots, i$. The final output of the algorithm is $G' = G_{m}$.

The way $E_0$ is processed is different from the way
$E_i$, $i > 0$ is processed. We now separately describe these two parts.

\subsection{Processing Edges in $E_0$}
\label{subsection:phase0}
%Let us call edges in $E_0$ {\em short}
%edges, and edges in $E \setminus E_0$ {\em long} edges. Then
%$G[E_0]$ is the spanning subgraph of $G$ containing just the short
%edges.
We start by stating a property of $G_0$
that follows easily from the fact that all edges in
$G_0$ are small.
\begin{lemma} Every connected component of $G_0$ induces a clique in $G$.
\label{lem:clique}
\end{lemma}
The algorithm {\tt PROCESS-SHORT-EDGES} for processing edges in
$E_0$ consists of three steps (i) determine the connected components
of $G_0$, (ii) use {\tt SEQ-GREEDY} to compute a $t$-spanner for
each connected component (that is, a clique), and (iii) let $G'_0$
be the union of the $t$-spanners computed in Step (2) and output
$G'_0$. The following theorem states the correctness of the {\tt
PROCESS-SHORT-EDGES} algorithm. Its proof follows easily from the
correctness of {\tt SEQ-GREEDY}.

\begin{theorem}
\label{theorem:phase0} $G'_0$ satisfies the following properties. (i)
For every edge $\{u, v\} \in E_0$, $G'_0$ contains a $uv$-path of
length at most $t \cdot |uv|$, (ii) $\Delta(G'_0) = O(1)$, and
(iii) $w(G'_0) = O(w(MST(G)))$.
\end{theorem}

\subsection{Processing Long Edges}
\label{subsection:longerEdges}

We now describe how edges in $E_i$ are processed, for $i > 0$. The
algorithm {\tt PROCESS-LONG-EDGES} has five steps:
(i) computing a cluster cover for $G'_{i-1}$,
(ii) selecting query edges in $E_i$, (iii) computing a cluster graph
$H_{i-1}$ for $G'_{i-1}$, (iv) answering shortest path queries for
the query edges selected in Step (ii), and (v) removing redundant
edges.
These steps are described in the next five subsections.

%We will show that it suffices to process only a small subset
%of edges in $E_i$ (which we call {\em query}~edges) to guarantee a
%spanner with all three desired properties.
%Properties $1$ through $3$ listed in Section~\ref{subsection:sequentialGreedyAlgorithm}.
%Prior to selecting query edges, we partition the vertices in $V$ into clusters
%using techniques similar to those described in~\cite{DasNarasimhan97}.

For any graph $J$, let $V(J)$ denote the vertex set for $J$. For any
pair of vertices $u, v \in V(J)$ let $\ssp_J(u, v)$ denote the
length of a shortest $uv$-path in $J$. Define a {\em cluster\/} of $J$
with {\em center\/} $u \in V(J)$ and {\em radius\/}
$r$ to be a set of vertices $C_u \subseteq V(J)$ such that,
for each $v \in C_u$, $\ssp_J(u, v) \le r$.
A set of clusters $\{C_{u_1}, C_{u_2}, \ldots\}$ of $J$ is a
{\em cluster cover of $J$ of radius $r$\/} if every cluster in the
set has radius $r$, every vertex in $V(J)$ belongs to at least
one cluster, and for any pair of cluster centers $u_i$ and $u_j$,
$\ssp_J(u_i, u_j) > r$.

\subsubsection{Computing a Cluster Cover for $G'_{i-1}$}
\label{sec:cluster.cover} At the beginning of phase $i$ we compute a
cluster cover of radius $\delta W_{i-1}$, where $\delta < 1$ is a
constant that will be fixed later. We start with an arbitrary vertex
$u \in V$ and run Dijkstra's shortest path algorithm with source $u$
on $G'_{i-1}$, in order to identify nodes $v \in V$ with the
property that $\ssp_{G'_{i-1}}(u, v) \le \delta W_{i-1}$; each such
node $v$ gets included in the cluster $C_u$. Once $C_u$ has been
identified, recurse on $V \setminus C_u$ until all nodes belong to
some cluster and we have a cluster cover of $G'_{i-1}$ of radius
$\delta W_{i-1}$.

\subsubsection{Selecting Query Edges in $E_i$}
\label{sec:query.edges} As defined earlier, edges in $E_i$ have
weights in the interval $I_i = (W_{i-1}, W_i]$, while the cluster
cover for $G'_{i-1}$ has radius $\delta W_{i-1}$, with $\delta < 1$.
This implies that each edge in $E_i$ has endpoints in different
clusters. Our goal is to select a unique query edge per pair of
clusters. This will guarantee that there are a constant number of
query edges incident on any node (see Lemma \ref{lem:query.edges})
and this fact will be critically used by the distributed version of
our algorithm to guarantee the degree bound on the spanner that is
constructed.

Let $\theta$ be a quantity that satisfies $0 < \theta <
\frac{\pi}{4}$ and $t \ge 1/(\cos\theta - \sin\theta)$. %\frac{1}{\cos\theta - \sin\theta}$.
For any value $t > 1$, no matter how small, there always exists a
$\theta$ that satisfies these restrictions. Define an edge $e = \{u,
v\} \in E_i$ to be a {\em covered edge\/} if there is a $z \in V$
such that (i) $\{u, z\} \in G'_{i-1}$, $|vz| \le \alpha$ and
$\angle{vuz} \le \theta$ or (ii) $\{v, z\} \in G'_{i-1}$, $|uz| \le
\alpha$ and $\angle{uvz} \le \theta$. Any edge in $E_i$ that is not
covered is a {\em candidate\/} query edge. The motivation for these
definitions is the following geometric lemma, due to Czumaj and Zhao
\cite{CzumajZhao}.

\begin{lemma}[Czumaj and Zhao~\cite{CzumajZhao}]
Let $0 < \theta < \frac{\pi}{4}$ and $t \ge \frac{1}{\cos\theta -
\sin\theta}$. Let $u, v, z$ be three points in $\R^d$ with
$\angle{vuz} \le \theta$. Suppose further that $|uz| \le |uv|$. Then
the edge $\{u, z\}$ followed by a $t$-spanner path from $z$ to $v$
is a $t$-spanner path from $u$ to $v$ (see Figure~\ref{fig:cz}).
\label{lem:edge.path}
\end{lemma}

\begin{figure}[htpb]
\centerline{\includegraphics[width =
0.40\linewidth]{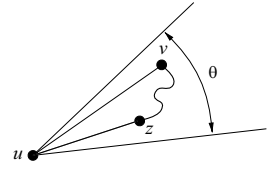}} \caption{(a) Edge $\{u, v\}$ is
{\em covered}: $\{u, z\}$ followed by a $t$-spanner $zv$-path is a
$t$-spanner $uv$-path.} \label{fig:cz}
\end{figure}

Now note that for each covered edge $\{u, v\} \in E_i$, there exists
$z$ that satisfies the preconditions of Lemma~\ref{lem:edge.path}
(by definition), and using this lemma we can show that $G'_{i-1}$
already contains a $uv$-path of length at most $t \cdot |uv|$. This
suggests that covered edges need not be queried and therefore we can
start with the complement of the set of covered edges as candidate
query edges.
%Using this lemma, it is easy to see that for any covered edge $\{x,
%y\} \in E_i$, $G'_{i-1}$ already contains an $xy$-path of length at
%most $t \cdot |xy|$ and hence covered edges need not be queried.
%Therefore, we take the complement of the set of covered edges and
%start with these as candidate query edges.

For each pair of clusters $C_a$ and $C_b$, let $E_i[C_a, C_b]$
denote the subset of candidate query edges in $E_i$ with one
endpoint in $C_a$ and the other endpoint in $C_b$. Our algorithm
selects a unique query edge $\{x, y\}$ from each nonempty subset
$E_i[C_a, C_b]$. Assuming that $x \in C_a$ and $y \in C_b$, the edge
$\{x, y\}$ is selected so as to minimize
\begin{equation}
t \cdot |xy| - \ssp_{G'_{i-1}}(a,x) - \ssp_{G'_{i-1}}(b, y)
\label{eq:query.pairs}
\end{equation}
The quantity in~(\ref{eq:query.pairs}) is carefully chosen to
guarantee that, if a $t$-spanner path between the endpoints of an
edge $\{x, y\}$ that minimizes~(\ref{eq:query.pairs}) exists in
$G'_i$, then $t$-spanner paths between the endpoints of {\em all}
edges in $E_i[C_a, C_b]$ exist in $G'_i$ (this property will later
be shown in the proof of Theorem~\ref{thm:relaxed.spanner}). This
implies that, for each pair of clusters $C_a$ and $C_b$,  querying
the edge $\{x, y\}$ in $E_i[C_a, C_b]$ that
minimizes~(\ref{eq:query.pairs}) renders querying any other edge in
$E_i[C_a, C_b]$ redundant.
%
%In Section~\ref{sec:properties} we show that it suffices to answer
%shortest path queries on the selected query edges to ensure that
%$G'_i$ is a $t$-spanner of $G_i$ at the end of phase $i$.

The following lemma shows that selecting query edges as described
above filters all but a constant number of edges per cluster. The
proof follows from two observations: (i) if a pair of cluster
centers are connected by an edge in $E_i$, then the clusters are not
too far from each other in Euclidean space (in particular, no
farther than $(4\delta+r)W_{i-1}$), and (ii) the Euclidean distance
between any pair of cluster centers is bounded from below by
$\delta W_{i-1}/t$, because they would otherwise be part of the
same cluster.
\begin{lemma}
The number of query edges in $E_i$ that are incident on any cluster
is $O(t^d(\frac{4\delta + r}{\delta})^d)$, a constant.
\label{lem:query.edges}
\end{lemma}

\subsubsection{Computing a Cluster Graph}
\label{sec:cluster.graph}

For each selected query edge $\{x, y\} \in E_i$, we need to know if
$G'_{i-1}$ contains an $xy$-path of length at most $t \cdot |xy|$.
In general, the number of hops in a shortest $xy$-path in $G'_{i-1}$
can be quite large and having to traverse such a path would mean
that the shortest path query corresponding to edge $\{x, y\}$ could
not be answered quickly enough. To get around this problem, we use
an idea from \cite{DasNarasimhan97} in which the authors construct
an approximation to $G'_{i-1}$, called a {\em cluster graph\/}, and
show that for any edge $\{x, y\} \in E_i$, the shortest path query
for $\{x, y\}$ can be answered approximately on $H_{i-1}$ in a
constant number of steps. The goal of Das and Narasimhan
\cite{DasNarasimhan97} was to improve the running time of {\tt
SEQ-GREEDY} on complete Euclidean graphs, but we show that the
Das-Narasimhan data structure can be constructed and maintained in a
distributed fashion for efficiently answering shortest path queries
for edges belonging to a $\alpha$-UBG. In the following, we describe
a sequential algorithm that starts with a cluster cover of
$G'_{i-1}$ of radius $\delta W_{i-1}$, and builds a {\em cluster
graph\/} $H_{i-1}$ of $G'_{i-1}$. This algorithm is identical to the
one in Das and Narasimhan \cite{DasNarasimhan97} and is included
mainly for completeness.

The vertex set of $H_{i-1}$ is $V$ and the edge set of $H_{i-1}$
contains two types of edges: {\em intra-cluster} edges and {\em
inter-cluster} edges. An edge $\{a, x\}$ is an intra-cluster edge if
$a$ is a cluster center and $x$ is node in $C_a$. Inter-cluster
edges are between cluster centers. An edge $\{a, b\}$ is an
inter-cluster edge if $a$ and $b$ are cluster centers, and at least
one of the following two conditions holds: (i) $\ssp_{G'_{i-1}}(a,
b) \le W_{i-1}$, or (ii) there is an edge in $G'_{i-1}$ with one
endpoint in $C_a$ and the other endpoint in $C_b$. See
Figure~\ref{fig:clustergraph}.
%
%Note that unlike in~\cite{DasNarasimhan97}, we do not add to $H$
%inter-cluster edges $(u, v)$ satisfying $d(u,v) \le W_i$.
%
\begin{figure}[htbp]
\centering
\includegraphics[width=0.75\linewidth]{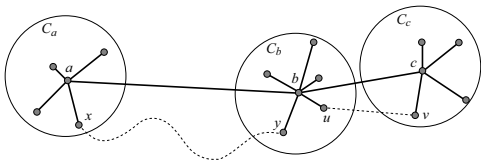}
\caption{Edges interior to disks are {\em intra-cluster} edges. Edge
$\{a, b\}$ is an {\em inter-cluster} edge because
$\ssp_{G'_{i-1}}(a, b) \le W_{i-1}$, and $\{b, c\}$ is an {\em
inter-cluster} edge because $\{u, v\}$ is in $G'_{i-1}$. An
$xy$-path in $G'_{i-1}$, shown by the dashed curve may be
approximated by the path $x, a, b, y$ in $H_{i-1}$.}
\label{fig:clustergraph}
\end{figure}

Regardless of the type of a cluster edge $e = \{a, b\}$ (inter- or
intra-), the weight of $e$ is the value of $sp_{G'_{i-1}}(a, b)$.
The following lemma
follows easily from the definition of inter-cluster edges.
%Note that, since $\delta < \alpha$, inter-cluster edges between a
%cluster center and its cluster nodes always exist.

\begin{lemma}
For any inter-cluster edge $\{a, b\}$ in $H_{i-1}$, we have that
$\ssp_{G'_{i-1}}(a, b) \le (2\delta+1)W_{i-1}$.
\label{lem:intercluster.distance}
\end{lemma}

The above upper bound also implies that $|ab| \le (2\delta +
1)W_{i-1}$. Using this and arguments similar to those used for Lemma
\ref{lem:query.edges}, we can show that the number of
inter-cluster edges incident to a cluster center is
$O((5+1/\delta)^d)$, so we have the following lemma.
\begin{lemma}
The number of inter-cluster edges in $H_{i-1}$ incident to a cluster
center is $O((5+1/\delta)^d)$, a constant. \label{lem:H.degree}
\end{lemma}
%\begin{proof}
%Fix an arbitrary cluster center $u$ and consider a sphere $S$ of
%radius $\delta W_{i-1} + (2 \delta+1)W_{i-1} + 2\delta W_{i-1}$
%centered at $u$. The first term $\delta W_{i-1}$ accounts for the
%radius of $C_u$, the second term $(2 \delta+1)W_{i-1}$ accounts for
%the maximum distance between the endpoints of an inter-cluster edge
%(cf.~Lemma~\ref{lem:intercluster.distance}), and the third term
%$2\delta W_{i-1}$ accounts for the diameter of a cluster. Then, for
%each inter-cluster edge $\{u, v\}$, $C_v$ must lie entirely in $S$.
%Since there are at most $O((5+\frac{1}{\delta})^2)$ clusters $C_v$
%in $S$ (cf. Lemma~\ref{lem:sphere.pack}), the number of
%inter-cluster edges incident to $u$ must be constant.
%\end{proof}

The main reason for constructing the cluster graph $H_{i-1}$ is that
lengths of paths in $H_{i-1}$ are close to lengths of corresponding
paths in $G'_{i-1}$ and shortest path queries for edges in $E_i$ can
be answered quickly in $H_{i-1}$. The following lemma (whose proof
appears in Das and Narasimhan \cite{DasNarasimhan97}) shows that we
can construct $H_{i-1}$ such that path lengths in $H_{i-1}$
approximate path lengths in $G'_{i-1}$ to any desired extent,
depending on the choice of $\delta$.

\begin{lemma}
For any edge $\{x, y\}\in E_i$, if there is a path between $x$ and
$y$ in $G'_{i-1}$ of length $L_1$, then there is a path between $x$
and $y$ in $H_{i-1}$ of length $L_2$ such that
$L_1 \le L_2 \le \frac{1 + 6\delta}{1-2\delta} L_1$.
%\frac{3\alpha-2\delta}{\alpha-2\delta} L_1\]
\label{lem:H.paths}
\end{lemma}

\subsubsection{Answering Shortest Path Queries}
\label{sec:query.answer} For query edges $\{x, y\} \in E_i$, we are
interested in knowing whether $G'_{i-1}$ has an $xy$-path of length
at most $t \cdot |xy|$. We ask this question on the cluster graph
$H_{i-1}$. If $H_{i-1}$ contains an $xy$-path of length at most $t
\cdot |xy|$, we do not add $\{x, y\}$ to $G'_i$; otherwise we do. If
$H_{i-1}$ contains an $xy$-path of length at most $t \cdot |xy|$,
then so does $G'_{i-1}$ (by Lemma \ref{lem:H.paths}, since $L_1 \le
L_2$). Therefore, not adding $\{x, y\}$ to the spanner is not a
dangerous choice. On the other hand, even if $H_{i-1}$ does not
contain an $xy$-path of length at most $t \cdot |xy|$, $G'_{i-1}$
might contain such a path and in this case adding edge $\{x, y\}$ is
unnecessary. Adding extra edges is of course not problematic for the
$t$-spanner property. It will turn out that this is not a problem
even for the requirement that the spanner should have bounded degree
and small weight, given that paths in $H_{i-1}$ can approximate
paths in $G'_{i-1}$ to an arbitrary degree.

Given the structure of the cluster graph, all but at most 2 edges in
any simple $xy$-path are inter-cluster edges. Since the radius of
each cluster is $\delta W_{i-1}$, each inter-cluster edge has weight
greater than $\delta W_{i-1}$. We are looking for a path of length
at most $t \cdot |xy|$. Since $|xy| \in (W_{i-1}, W_i]$, we are
looking for a path of length at most $t \cdot W_i = t \cdot r \cdot
W_{i-1}$. Any simple path in $H_{i-1}$ of length at most $t \cdot r
\cdot W_{i-1}$ has at most 2 + $\lceil tr/\delta \rceil$ hops, which
is a constant. This yields the following lemma.

\begin{lemma}
\label{lemma:constantH} For any edge $\{x, y\} \in E_i$, if
$\ssp_{H_{i-1}}(x, y) \le t \cdot |xy|$, then $H_{i-1}$ contains a
shortest $xy$-path with $O(1)$ hops (no more than 2 + $\lceil
tr/\delta \rceil$).
\end{lemma}

One issue we need to deal with, especially when attempting to
construct and answer queries in $H_{i-1}$ in a distributed setting,
is that edges in $H_{i-1}$ need not be present in the underlying
network $G$. Specifically, for an intra-cluster edge $\{u, a\}$,
where $C_a$ is a cluster and $u \in C_a$, it may be the case that
$|ua| > \alpha$ and $\{u, a\}$ may be absent from $G$. Similarly, an
inter-cluster edge $\{a, b\}$ in $H_{i-1}$ may be absent in $G$.
However, for any edge $\{x, y\}$ in $H_{i-1}$ (intra- or
inter-cluster edge), we have the bound $\ssp_{G'_{i-1}}(x, y) \le
(2\delta + 1) W_{i-1}$. This follows from Lemma
\ref{lem:intercluster.distance} and the fact that the radius of each
cluster is $\delta W_{i-1}$. Thus a shortest $xy$-path in $G'_{i-1}$
lies entirely in a ball of radius $(2\delta + 1) W_{i-1}$ centered
at $x$. Since $G'_{i-1}$ is a spanning subgraph of $G$, this implies
that there is a shortest $xy$-path $P$ in $G$ that lies entirely in
the $d$-dimensional ball of radius $(2\delta + 1) W_{i-1}$ centered
at $x$. Since any two vertices in $P$ that are two hops away from
each other are at least $\alpha$ apart (in the $d$-dimensional
Euclidean space), $P$ contains at most $\lceil 2 (2\delta + 1)
W_{i-1}/\alpha\rceil < \lceil 2(2 \delta + 1)/\alpha \rceil$ hops.
This argument
% along with Lemma \ref{lemma:constantH} -- can't see how it's used
yields the following theorem.

\begin{theorem}
For any edge $\{x, y\} \in E_i$, if $\ssp_{H_{i-1}}(x, y) \le t
\cdot |xy|$, then $G$ contains a shortest $xy$-path with $O(1)$ hops
(no more than $\lceil 2(2 \delta + 1)/\alpha \rceil$).
\end{theorem}

This theorem implies that brute force search initiated from one of
the endpoints, say $x$, will be able to answer the shortest path
query on edge $\{x, y\}$ in $O(1)$ rounds in a distributed setting.

\subsubsection{Removing Redundant Edges} \label{sec:redundant.edges}

Let $t_1$ be such that $1 < t_1 < t$. Recall that shortest path
queries for edges in $E_i$ are answered on $H_{i-1}$, and so updates
to $G'_i$ in phase $i$ do not influence subsequent shortest path
queries in phase $i$. Thus it is possible that in phase $i$ two
edges $\{u, v\}$ and $\{u', v'\}$ get added to $G_i$, yet both of
the following hold:
%(refer to Figure~\ref{fig:redundant}):
\begin{itemize}
\item [(i)] $\ssp_{H_{i-1}}(u, u') + |u'v'| + \ssp_{H_{i-1}}(v', v) \le t_1 \cdot
|uv|$
\item [(ii)] $\ssp_{H_{i-1}}(u', u) + |uv| +
               \ssp_{H_{i-1}}(v, v') \le t_1 \cdot |u'v'|$
\end{itemize}
Note that, since $\ssp_{G'_{i-1}}(x, y) \le \ssp_{H_{i-1}}(x, y)$
holds for any pair of nodes $x$ and $y$, and since $t_1 < t$,
conditions (i) and (ii) above imply that $G'_i$ contains $t$-spanner
paths from $u$ to $v$ and from $u'$ to $v'$. We call two edges $\{u,
v\}$ and $\{u', v'\}$ satisfying conditions (i) and (ii) above {\em
mutually redundant}: one of them could potentially be eliminated
from $G_i$, without compromising the $t$-spanner property of $G_i$.
In fact, such mutually redundant pairs of edges need to be
eliminated from $G'_i$ because our proof that $G'$ has small weight
(Theorem \ref{thm:boundedweight}) depends on the absence of such
pairs of edges.

To do this,
we build a graph $J$ that has a node for each edge in a mutually
redundant pair and an edge
between every pair of nodes that correspond to a mutually redundant
pair of edges in $G'_i$.
We construct an MIS $I$ of $J$ and eliminate from $G'_i$
all edges associated with nodes in $J$ that do not appear in $I$.

\subsection{The Three Desired Properties}
\label{sec:properties}

Let $G' = G'_{m}$ be the spanner at the end of phase $m$. We now
prove that $G'$ satisfies the three properties that the output of
{\tt SEQ-GREEDY} was guaranteed to have. The proofs of these
theorems form the technical core of the paper and are presented next
in this section.

\begin{theorem}
For any $0 < \delta \le \frac{t - t_1}{4}$, the output $G'$ is a
$t$-spanner. \label{thm:relaxed.spanner}
\end{theorem}
\begin{proof}
We first prove that the theorem holds for all query edges in $E$,
then we extend the argument to non-query edges as well. Let $\{x,
y\}$ be an arbitrary query edge and let $i \ge 1$ be such that $\{x,
y\} \in E_i$. Then either (i) $\{x, y\}$ is added to the spanner in
phase $i$, or (ii) $\ssp_{H_{i-1}}(x, y) \le t \cdot |xy|$. If the
former is true and $\{x, y\}$ is not a redundant edge, then the
theorem holds. If $\{x, y\}$ is a redundant edge but does not get
removed from $G_i$, then again the theorem holds. If $\{x, y\}$ is a
redundant edge that gets removed from $G_i$, then at least one
mutually redundant counterpart edge must remain in $G_i$ (since
removed edges form an independent set), ensuring a $t$-spanner
$xy$-path in $G_i$. If (ii) is true, then from
Lemma~\ref{lem:H.paths}, $\ssp_{G'_{i-1}}(x, y) \le
\ssp_{H_{i-1}}(x, y)$ (first part of the inequality) and therefore
$\ssp_{G'_{i-1}}(x, y) \le t \cdot |xy|$.

For non-query edges, the proof is by induction on the length of
edges in $G$. The base case corresponds to edges in $E_0$, for which
{\tt SEQ-GREEDY} ensures that the theorem holds.

Assume that the theorem is true for any edge in $E$ of length no
greater than some value $q$, and consider a smallest non-query edge
$\{x, y\}$ in $G$ of length greater than $q$. We prove that
$\ssp_{G'}(x, y) \le t \cdot |xy|$. Let $i$ be such that $\{x, y\}
\in E_i$. We now consider two cases, depending on whether $\{x, y\}$
is a {\em candidate\/} query edge in phase $i$ or not.

If $\{x, y\}$ is not a candidate query edge, then it is a covered
edge. That is, there exists an edge $\{x, z\}$ in $G'_{i-1}$ such
that $|yz| \le \alpha$ and $\angle{yxz} \le \theta$, or an edge
$\{y, z\}$ in $G'_{i-1}$ such that $|xz| \le \alpha$ and
$\angle{xyz} \le \theta$. The two cases are symmetric and so without
loss of generality, assume that the former is true. Here $\theta$
satisfies the hypothesis of the Czumaj-Zhao lemma (Lemma
\ref{lem:edge.path}), that is, $0 < \theta < \frac{\pi}{4}$ and $t
\ge \frac{1}{\cos\theta - \sin\theta}$. Since $|yz| \le \alpha$ and
$G$ is an $\alpha$-UBG, this implies that $\{y, z\}$ is an edge is
$E$. Furthermore, since $0 < \theta < \frac{\pi}{4}$, we have $|yz|
< |xy|$. Refer to Figure~\ref{fig:spanner.proof}a. If $\{y, z\}$ is
a query edge, then by the argument above we have that $G'$ contains
a $t$-spanner $yz$-path $p$. Otherwise, if $\{y, z\}$ is not a query
edge, since its length is less than the length of $\{x, y\}$, by the
inductive hypothesis we get that there is a $t$-spanner $yz$-path
$p$. In either case, Lemma~\ref{lem:edge.path} tells us that $\{x,
z\}$ followed by $p$ is a $t$-spanner path from $x$ to $y$,
completing this case.

\begin{figure}[htbp]
\centering
\begin{tabular}{c@{\hspace{0.2\linewidth}}c}
\includegraphics[width=0.2\linewidth]{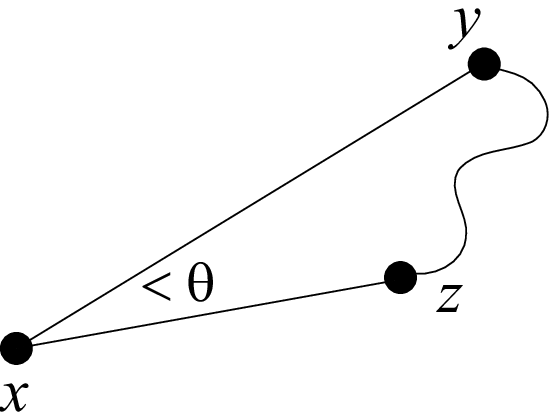} &
\includegraphics[width=0.38\linewidth]{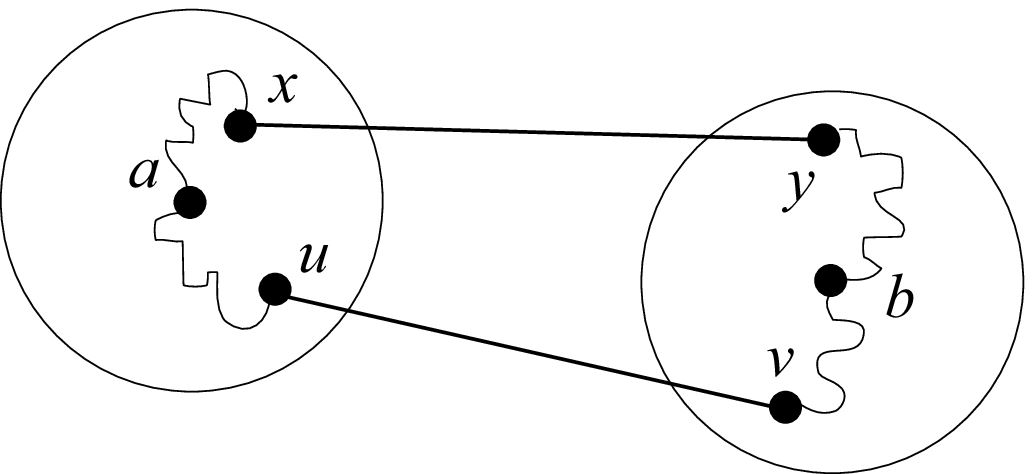} \\
(a) & (b)
\end{tabular}
\caption{(a) $\{x, y\}$ is a covered edge (b) $\{u, v\}$ is a query
edge: if $G_i$ contains a $t$-spanner $uv$-path, then $G_i$ contains
a $t$-spanner $xy$-path.} \label{fig:spanner.proof}
\end{figure}

We now consider the case when $\{x, y\}$ is a candidate query edge
in phase $i$, but not a query edge. Let $a$ and $b$ be such that $x
\in C_a$ and $y \in C_b$, and let $\{u, v\}$ be the query edge
selected in phase $i$, with $u \in C_a$ and $v \in C_b$. Refer to
Figure~\ref{fig:spanner.proof}b. Due to the criteria for selecting
$\{u, v\}$, we have
%\begin{equation}
%t \cdot |uv| - \ssp_{G'_{i-1}}(a, u) - \ssp_{G'_{i-1}}(b, v) \le t
%\cdot |xy| - \ssp_{G'_{i-1}}(a, x) - \ssp_{G'_{i-1}}(b, y)
%\label{eq:spannerproof.query}.
%\end{equation}

\begin{eqnarray}
t \cdot |uv| - \ssp_{G'_{i-1}}(a, u) - \ssp_{G'_{i-1}}(b, v) \le
\nonumber
\\
t \cdot |xy| - \ssp_{G'_{i-1}}(a, x) - \ssp_{G'_{i-1}}(b, y)
\label{eq:spannerproof.query}.
\end{eqnarray}

Recall that $G'_i$ is the partial spanner at the end of phase $i$.
We show that $\ssp_{G'_i}(x, y) \le t \cdot |xy|$. We discuss two
cases, depending on whether $\{u, v\}$ was added to $G'_i$ or not.

Assume first that $\{u, v\}$ was not added to $G'_i$. This means
that $\ssp_{H_{i-1}}(u, v) \le t \cdot |uv|$. Note however that
%\begin{equation}
%\ssp_{H_{i-1}}(u, v) = \ssp_{G'_{i-1}}(u, a) + \ssp_{H_{i-1}}(a, b) +
%\ssp_{G'_{i-1}}(b, v) \le t \cdot |uv| \label{eq:spannerproof0}.
%\end{equation}
\begin{eqnarray}
\ssp_{H_{i-1}}(u, v) & = & \ssp_{G'_{i-1}}(u, a) + \ssp_{H_{i-1}}(a,
b)
+ \ssp_{G'_{i-1}}(b, v) \nonumber \\
& \le & t \cdot |uv| \label{eq:spannerproof0}.
\end{eqnarray}
We now evaluate
\begin{eqnarray*}
       \ssp_{G'_{i-1}}(x, y) & \le &
              \ssp_{G'_{i-1}}(x, a) + \ssp_{G'_{i-1}}(a, b) +
                  \ssp_{G'_{i-1}}(b, y) \\
       & \le &  \ssp_{G'_{i-1}}(x, a) + \ssp_{H_{i-1}}(a, b) +
                  \ssp_{G'_{i-1}}(b, y) \\
       & \le & t \cdot |xy|.
\end{eqnarray*}
This latter inequality involves simple substitutions that use
inequalities~(\ref{eq:spannerproof.query})
and~(\ref{eq:spannerproof0}), and completes this case.

\smallskip
\noindent
Now assume that $\{u, v\}$ was added to $G'_i$. Since $u \in
C_a$ and $C_a$ has radius $\delta W_{i-1}$, we have that
$\ssp_{G'_{i-1}}(a, u) \le \delta W_{i-1}$. Similarly,
$\ssp_{G'_{i-1}}(b, v) \le \delta W_{i-1}$. These together
with~(\ref{eq:spannerproof.query}) yield
\begin{equation}
t \cdot |uv| - 2\delta W_{i-1} \le t \cdot |xy| -
\ssp_{G'_{i-1}}(a, x) - \ssp_{G'_{i-1}}(b, y)
\label{eq:spannerproof1}.
\end{equation}
If the edge $\{u, v\}$ turns out to be redundant and eliminated from
$G_i$, the existence of a mutually redundant counterpart edge in
$G'_i$ ensures that $\ssp_{G'_i} (u, v) \le t_1 \cdot |uv|$.
This %The existence of $\{u, v\}$ in $G'_i$
enables us to construct in $G'_i$ a path from $a$ to $b$ of weight

\begin{eqnarray}
        \ssp_{G'_i}(a, b) & \le & \ssp_{G'_i}(a, u) + t_1 \cdot |uv| + \ssp_{G'_i}(v, b) \nonumber\\
          & \le & 2 \delta W_{i-1} + t_1 \cdot |uv|,
\label{eq:spannerproof2}
\end{eqnarray}
since $\ssp_{G'_i}(a, u) \le \ssp_{G'_{i-1}}(a, u) \le \delta
W_{i-1}$, and same for $\ssp_{G'_i}(v, b)$. We can now construct a
path in $G'_i$ from $x$ to $y$ of weight
\begin{tabbing}
\=..................\=........\=
......................................................\=...\kill
       \> $\ssp_{G'_i}(x, y)$ \> $\le$
              \> $\ssp_{G'_i}(a, x) + \ssp_{G'_i}(b, y) + \ssp_{G'_i}(a, b) $\\
       \>   \> $\le$ \> $t \cdot |xy| + 2\delta W_{i-1} - t \cdot |u v| + \ssp_{G'_i}(a, b)$\\
%       \> (...substituting~(\ref{eq:spannerproof1}))\\
       \>   \> $\le$ \> $t \cdot |x y| + 4\delta W_{i-1} - (t-t_1) \cdot |u v|$ \\
%       \> (...substituting~(\ref{eq:spannerproof2})) \\
       \>   \> $<$ \> $t \cdot |x y| + 4\delta W_{i-1} - (t-t_1)W_{i-1}$
%       \>(...since $|u v| > W_{i-1}$)
\end{tabbing}
%\begin{eqnarray*}
%       \ssp_{G_i}(x, y) & \le
%              & \ssp_{G_i}(a, x) + \ssp_{G_i}(b, y) + \ssp_{G_i}(a, b) \\
%       & \le & t \cdot d(x, y) + 2\delta W_{i-1} - t \cdot d(u, v) + \ssp_{G_i}(a, b) (\mbox{using~(\ref{eq:spannerproof1})})\\
%       & \le & t \cdot d(x, y) + 4\delta W_{i-1} - (t-1) \cdot d(u, v)
%       (\mbox{using~(\ref{eq:spannerproof2})}) \\
%       & < & t \cdot d(x, y) + 4\delta W_{i-1} - (t-1) \cdot d(u,
%       v) (\mbox{since~} d(u, v) > W_{i-1})
%\end{eqnarray*}
In deriving this chain of inequalities, we have
used~(\ref{eq:spannerproof1}),~(\ref{eq:spannerproof2}) and the fact
that $|u v| > W_{i-1}$. Note that for any $\delta \le
\frac{t-t_1}{4}$, the quantity $4\delta W_{i-1} - (t-1) \cdot
W_{i-1}$ above is negative, yielding $\ssp_{G_i}(x, y) < t \cdot |x
y|$. This completes the proof.
\end{proof}

\begin{figure}[htpb]
\centerline{
\begin{tabular}{cc}
\includegraphics[width = 0.32\linewidth]{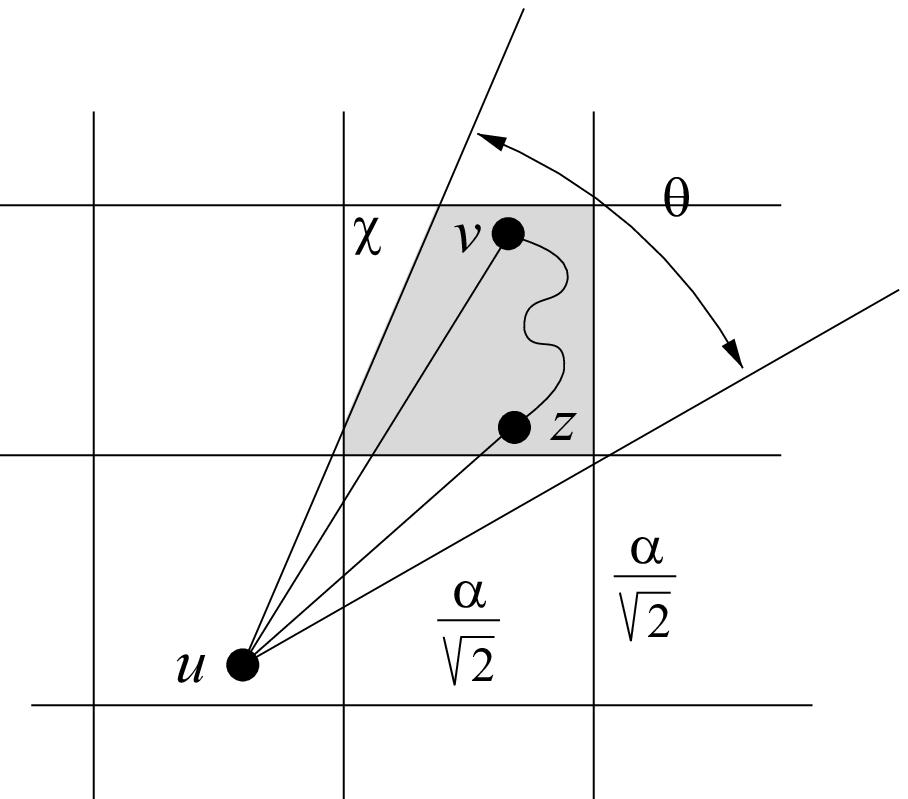} &
\raisebox{2ex}{\includegraphics[width =
0.4\linewidth]{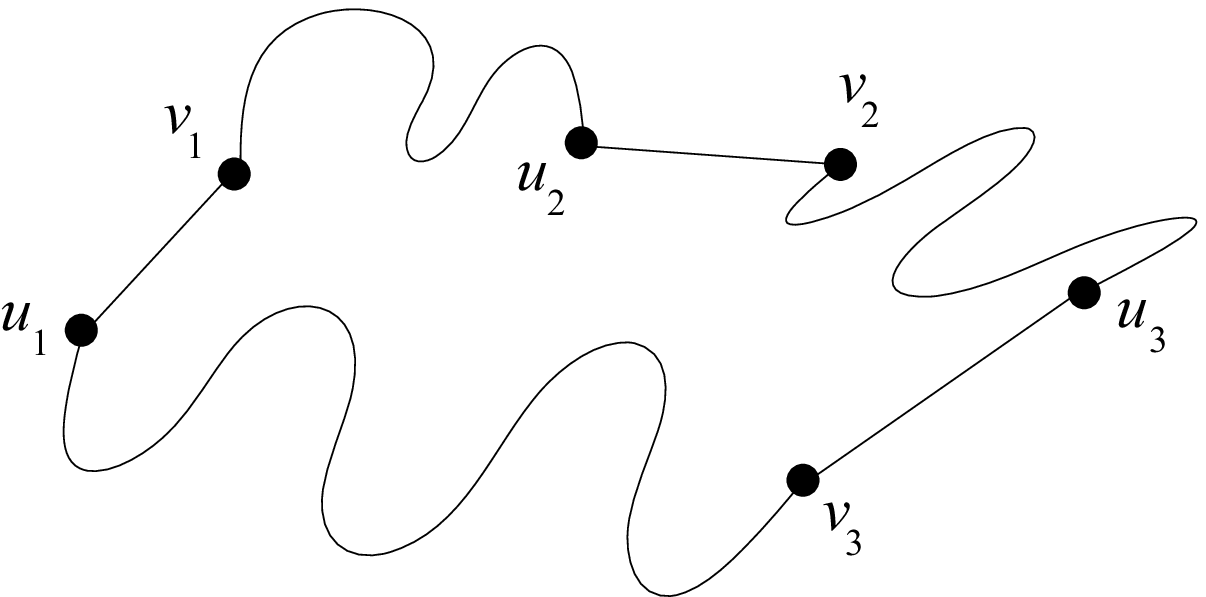}}
\end{tabular}}
\caption{(a) Region $\chi$ contains two neighbors $v$ and $z$ of
$u$. (b) Definition of the $t$-leapfrog property with $S = \{\{u_1,
v_1\}, \{u_2, v_2\}, \{u_3, v_3\}\}$.} \label{fig:sequentialGreedy}
\end{figure}

\begin{theorem}
$G'$ has $O(1)$ degree.
\label{thm:constant.degree}
\end{theorem}
\begin{proof}
Let $\theta$ be a quantity satisfying the conditions of
Lemma~\ref{lem:edge.path}. Fix a vertex $u$ and consider the
$d$-dimensional unit radius ball centered at $u$. For some $T$ that
depends only on $\theta$ and $d$, this ball can be partitioned into
$T$ cones, each with apex $u$, such that for any $x$, $y$ in a cone,
$\angle xuy \le \theta$.
Yao~\cite{Yao82} shows how to construct such a partition with $T =
O(d^{3/2}\cdot \sin^{-d}(\theta/2)\cdot \log(d
\sin^{-1}(\theta/2)))$ cones.
Place an infinite axis-parallel grid of $d$-dimensional cubes, each
of dimension $\frac{\alpha}{\sqrt{d}} \times \frac{\alpha}{\sqrt{d}}
\times \cdots \times \frac{\alpha}{\sqrt{d}}$, on the plane. See
Figure~\ref{fig:sequentialGreedy}(c) for a 2-dimensional version of
this picture. There are $O(1/\alpha^d)$ cells that intersect the
unit ball centered at $u$, and therefore there are $O(1/\alpha^d)$
cells that intersect each cone in the cone partition of this unit
ball. Thus the cones and the square cells together partition the
unit ball centered at $u$ into $O(T/\alpha^d)$ regions. We show that
in $G'$, $u$ has %a constant number of neighbors in each region.
$O(\frac{t^d(4\delta+r)^d}{\delta^d})$ neighbors in each region,
which is a constant.

Let $v_1, v_2, \ldots, v_k$ be neighbors of $u$ in $G'$ that lie in
a region $\chi$. Without loss of generality, assume that $|uv_1| \ge
|uv_j|$, for $j=2, \ldots, k$, and let $i$ be such that $\{u, v_1\}
\in E_i$. Since $|uv_j| \le |uv_1|$, we have that for all $j = 2,
\ldots, k$, $\{u, v_j\} \in E_{\ell}$, with $\ell \le i$.

We now prove that $\{u, v_j\}$ is in fact in $E_i$ for all $j$. To
derive a contradiction, assume that there is a $j > 1$ such that
$\{u, v_j\} \in E_{\ell}$, with $\ell < i$. This means that just
before edge $\{u, v_1\}$ is processed, $G'$ contains edge $\{u,
v_j\}$. Also note that since $v_1$ and $v_j$ lie in the same region,
$|v_1v_j| \le \alpha$. But, this means that $\{u, v_1\}$ is a
covered edge in phase $i$ and will not be queried. This contradicts
the presence of edge $\{u, v_1\}$ in $G'$.

We have shown that $\{u, v_j\} \in E_i$ for all $j$. Recall that our
algorithm picks a unique query edge per pair of clusters. This along
with Lemma~\ref{lem:query.edges} proves that $k$ is constant.
\end{proof}

\smallskip
\noindent In the next theorem, we show that the spanner produced by
the algorithm has small weight. The proof relies on the line
segments in the spanner satisfying a property known as the {\em
leapfrog property}
\cite{CzumajZhao,GudmundssonLevcopoulosNarasimhan}. For any $t \ge
t_2 > 1$, a set of line segments, denoted $F$, has the {\em $(t_2,
t)$-leapfrog property\/} if for every subset $S = \{\{u_1, v_1\},
\{u_2, v_2\}, \ldots, \{u_s, v_s\}\}$ of $F$
\begin{equation}
t_2 \cdot |u_1 v_1| < \sum_{i=2}^s |u_i v_i| + t \cdot
\Big(\sum_{i=1}^{s-1} |v_i u_{i+1}| + |v_s u_1|\Big).
\label{eq:leapfrog}
\end{equation}
Informally, this definition says that if there exists an edge
between $u_1$ and $v_1$, then any path not including $\{u_1, v_1\}$
must have length greater than $t_2 |u_1 v_1|$ (see Figure
\ref{fig:sequentialGreedy}(c) for an illustration of this
definition). The following implication of the $(t_2, t)$-leapfrog
property was shown by Das and Narasimhan \cite{DasNarasimhan97}.

\begin{lemma}
Let $~t \ge t_2 > 1$. If the line segments $F$ in $d$-dimensional
space satisfy the $(t_2, t)$-leapfrog property, then $wt(F) =
O(wt(MST))$, where $MST$ is a minimum spanning tree connecting the
endpoints of line segments in $F$. The constant in the asymptotic
notation depends on $t$, $t_2$ and $d$.
\end{lemma}

\begin{theorem}
Let $0< \delta < \min\{(t-1)/(6+2t), (t-t_1)/4\}$. Let $t_\delta$
denote $t_1 \cdot (1-2\delta)/(1+6\delta)$. Let $1 < r < (t_\delta +
1)/2$. When the relaxed  greedy algorithm is run with these values
of $\delta$ and $r$, the output $G'$ satisfies $w(G') =
O(wt(MST(G)))$. \label{thm:boundedweight}
\end{theorem}
\begin{proof}
Let $\beta > 1$ be a constant picked as follows. When $t\alpha < 1$,
pick $\beta$ satisfying $1 < \beta < \min\{2, 1/(1-t\alpha)\}$.
Otherwise, pick $\beta$ satisfying $1 < \beta < 2$. Partition the
edges of $G'$ into subsets $F_0, F_1, \ldots$ such that $F_0 =
\{\{u, v\} \in G' \mid |uv| \le \alpha\}$ and for each $j > 0$, $F_j
= \{\{u, v\} \in G' \mid \alpha \beta^{j-1} < |uv| \le \alpha
\beta^j\}$. Let $\ell = \lceil \log_\beta \frac{1}{\alpha} \rceil$.
Then every edge in $G'$ is in some subset $F_j$, $0 \le j \le \ell$.
We will now show that each $F_j$ satisfies the $(t_2, t)$-leapfrog
property, for any $t_2$ satisfying:
\begin{equation}
1 \le t_2 < \min\{\frac{t_\delta + 1}{r} - 1, \frac{2}{r},
\frac{t}{r}, \frac{2}{\beta}, t\alpha + \frac{1}{\beta}\}.
\label{eq:conditions}
\end{equation}
It is easy to check that our choice for $\delta$, $r$, and $\beta$
guarantee that each quantity inside the min operator is strictly
greater than 1. Showing the $(t_2, t)$-leapfrog property for $F_j$
would imply that $w(F_j) = O(w(MST(G)))$, and since the edges of
$G'$ are partitioned into a constant number of subsets $F_j$, $w(G')
= O(w(MST(G)))$.

Consider an arbitrary subset $S = \{\{u_1, v_1\}, \{u_2, v_2\},
\ldots,$  $\{u_s, v_s\}\} \subseteq F_0$ . To prove
inequality~(\ref{eq:leapfrog}) for $S$, it suffices to consider the
case when $\{u_1, v_1\}$ is a longest edge in $S$. We consider $F_0$
separately from $F_j$, $j > 0$.

\smallskip
\noindent {\bf The $F_0$ case.} If for any $1 \le k < s$, $|v_k
u_{k+1}|
> |u_1 v_1|$ or $|v_s u_1|
> |u_1 v_1|$, then the leapfrog property holds. So we assume that
for all $1 \le k < s$, $|v_k u_{k+1}| \le |u_1 v_1|$ and $|v_s u_1|
\le |u_1 v_1|$. Let $i$ be the phase in which $\{u_1, v_1\}$ gets
processed, i.e., $\{u_1, v_1\} \in E_i$. Since $|u_1 v_1| \le
\alpha$, it is the case that for all $1 \le k < s$, $|v_k u_{k+1}|
\le \alpha$ and $|v_s u_1| \le \alpha$. Hence, $\{\{v_s, u_1\} \}
\cup \{\{v_k, u_{k+1}\} \mid 1 \le k < s\}$ is a subset of edges of
$G$ and each edge in this set gets processed in phase $i$ or
earlier.

Assume first that at least one edge in the set $\{\{v_s, u_1\} \}
\cup \{\{v_k, u_{k+1}\} \mid 1 \le k < s\}$ gets processed in phase
$i$. Then the right hand side of inequality~(\ref{eq:leapfrog}) is
at least $t W_{i-1}$, since edges in $E_i$ have weights in the
interval $I_i = (W_{i-1}, r W_{i-1}]$. Also since $t_2|u_1 v_1| \le
t_2 r W_{i-1}$, and since the inequality $ t_2 r W_{i-1} < t
W_{i-1}$ is guaranteed by the values of $r$ and $t_2$
in~(\ref{eq:conditions}), the leapfrog property holds for this case.

Assume now that all edges in $\{\{v_s, u_1\} \} \cup \{\{v_k,
u_{k+1}\} \mid 1 \le k < s\}$ have been processed in phase $i-1$ or
earlier, meaning that $t$-spanner paths between their endpoints
exist in $G'_{i-1}$ at the time $\{u_1, v_1\}$ gets processed. For
$1 \le k < s$, let $P_k$ be a shortest $v_k u_{k+1}$-path in
$G'_{i-1}$, and let $P_s$ be a shortest $v_s u_1$-path in
$G'_{i-1}$. Let $P$ be the following $u_1 v_1$-path in $G'_i$: $P =
P_1 \oplus \{u_2, v_2\} \oplus P_2 \oplus \{u_3, v_3\} \oplus \cdots
\oplus P_s$. Here, we use $\oplus$ to denote concatenation. We
distinguish three cases, depending on the size of the subset $S \cap
E_i$.
\begin{itemize}
\item[(i)] $|S \cap E_i| > 2$. Then, $w(P) \ge 2 W_{i-1}$. We also have
that $|u_1 v_1| \le r W_{i-1}$, since $\{u_1, v_1\} \in E_i$. It
follows that $w(P) > t_2 |u_1 v_1|$ for any $t_2 < \frac{2}{r}$.
Furthermore, $w(P)$ is no greater than the right hand side of the
$(t_2,t)$-leapfrog inequality~(\ref{eq:leapfrog}), so lemma holds
for this case as well.

\item[(ii)] $|S \cap E_i| = 2$.
In addition to $\{u_1, v_1\}$, assume that $\{u_k, v_k\} \in E_i$
for some $k$, $1 < k \le s$. It the $(t_2,t)$-leapfrog
inequality~(\ref{eq:leapfrog}) holds, we are done and so let us
assume the opposite of that:
\begin{equation}
t_2 \cdot |u_1 v_1| \ge \sum_{i=2}^s |u_i v_i| + t \cdot
\Big(\sum_{i=1}^{s-1} |v_i u_{i+1}| + |v_s u_1|\Big).
\label{eq:leapfrogOpposite}
\end{equation}
Since all edges $\{u_j, v_j\}$, $1 \le j \le s$, except for $\{u_1,
v_1\}$ and $\{u_k, v_k\}$ are in $G'_{i-1}$, and since $G'_{i-1}$
contains $t$-spanner $v_j u_{j+1}$-paths for all $j$, $1 \le j < s$,
and a $t$-spannner $v_s u_1$-path, the above inequality yields
$$t_2 \cdot |u_1 v_1| \ge \ssp_{G'_{i-1}}(v_1, u_k) + |u_k v_k| + \ssp_{G'_{i-1}}(v_k, u_1).$$
Multiplying both sides by $(1+6\delta)/(1-2\delta)$ and using $t_2 <
t_\delta$ (which is implied by our choice of $t_2$) and Lemma
\ref{lem:H.paths}, we get
\begin{equation}
t_1 \cdot |u_1 v_1| \ge \ssp_{H_{i-1}}(v_1, u_k) + |u_k v_k| +
\ssp_{H_{i-1}}(v_k, u_1). \label{eq:mutuallyRedundant1}
\end{equation}

Let $\Delta = \sum_{i=1}^{s-1} |v_i u_{i+1}| + |v_s u_1|$. We now
observe that %the inequality
\begin{equation}
t_\delta\cdot |u_k v_k | < \sum_{i=1}^{k-1} |u_i v_i| +
\sum_{i=k+1}^s |u_i v_i| + t \cdot \Delta \label{eq:oneStep}
\end{equation}
implies the $(t_2, t)$-leapfrog property. To see this use the fact
that both $\{u_1, v_1\}$ and $\{u_k, v_k\}$ belong to $E_i$ and
therefore $|u_1 v_1| < r\cdot |u_k v_k|$, which substituted
in~(\ref{eq:oneStep}) yields:
$$t_\delta\cdot |u_k v_k | - (r-1)\cdot |u_k v_k| < \sum_{i=2}^{s} |u_i v_i| +
t \cdot \Delta.$$ We get the lower bound $t_2 \cdot |u_1 v_1|$ on
the left hand side of the above inequality by using $|u_k v_k| >
|u_1 v_1|/r$ again and our choice of $t_2 < (t_\delta + 1)/r - 1$.
This yields the $(t_2, t)$-leapfrog property. So we assume that
inequality (\ref{eq:oneStep}) does not hold, that is,
$$t_\delta\cdot |u_k v_k | \ge \sum_{i=1}^{k-1} |u_i v_i| + \sum_{i=k+1}^s |u_i v_i| +
t \cdot \Delta.$$ Since all edges $\{u_j, v_j\}$, $1 \le j \le s$,
except for $\{u_1, v_1\}$ and $\{u_k, v_k\}$ are in $G'_{i-1}$, and
since $G'_{i-1}$ contains $t$-spanner $v_j u_{j+1}$-paths for all
$j$, $1 \le j < s$, and a $t$-spannner $v_s u_1$-path, the above
inequality yields
$$t_\delta\cdot |u_k v_k | \ge  \ssp_{G'_{i-1}}(v_1, u_k) + |u_1 v_1| + \ssp_{G'_{i-1}}(v_k, u_1).$$
Multiplying both sides by $(1+6\delta)/(1-2\delta)$ and using
Lemma \ref{lem:H.paths}, we get
\begin{equation}
t_1 \cdot |u_k v_k| \ge \ssp_{H_{i-1}}(v_1, u_k) + |u_1 v_1| +
\ssp_{H_{i-1}}(v_k, u_1). \label{eq:mutuallyRedundant2}
\end{equation}
Inequalities (\ref{eq:mutuallyRedundant1}) and (\ref{eq:mutuallyRedundant2}) imply
that edges $\{u_1, v_1\}$ and $\{u_2, v_2\}$ are mutually redundant
and therefore cannot both exist in the spanner --- a contradiction.
\item[(iii)] $|S \cap E_i| = 1$. This means that $P$ exists in
$G'_{i-1}$ at the time $\{u_1, v_1\}$ is processed. Furthermore,
$w(P)> t \cdot |u_1 v_1| > t_2 \cdot |u_1 v_1|$, otherwise $\{u_1,
v_1\}$ would not have been added to the spanner, a contradiction.
\end{itemize}

\noindent {\bf The $F_j$ case, $j > 0$.} In this case, $|u_k v_k|
> |u_1 v_1|/\beta$ for all $k = 2, 3, \ldots, s$. If $|S| \ge 3$,
then the right hand side of the $(t_2, t)$-leapfrog
inequality~(\ref{eq:leapfrog}) is at least $2 \cdot |u_1 v_1|/\beta$
and therefore the $(t_2, t)$-leapfrog inequality goes through for
any $1 < t_2 < 2/\beta$. Otherwise, if $|S| = 2$, then we need to
show that $t_2 \cdot |u_1 v_1| < |u_2 v_2| + t \cdot (|u_1 v_2| +
|u_2 v_1|)$. If each of $|u_1 v_2|$ and $|u_2 v_1|$ is at most
$\alpha$, then using the same argument as in the $F_0$-case with $|S
\cap E_i| = 2$, we can show that $\{u_1, v_1\}$ and $\{u_2, v_2\}$
are mutually redundant and will not both exist in the spanner.
Otherwise, if one of $|u_1 v_2|$ or $|u_2 v_1|$ is greater than
$\alpha$, then the right hand side of the $(t_2, t)$-leapfrog
inequality~(\ref{eq:leapfrog}) is greater than $|u_1 v_1|/\beta + t
\alpha$. To ensure that the inequality goes through, we require that
$t_2 \cdot |u_1 v_1| \le \frac{|u_1 v_1|}{\beta} + t \alpha$. Since
$|u_1 v_1| \le 1$, the above inequality is satisfied for any $1 <
t_2 \le t\alpha + \frac{1}{\beta}$, which holds true
cf.~(\ref{eq:conditions}).
%If $t \alpha \ge 1$, this
%inequality can easily be satisfied. If $t \alpha < 1$, then we can
%make the inequality hold by picking $1 < \beta < \frac{1}{(1 - t
%\alpha)}$, thus satisfying $t\alpha + \frac{1}{\beta} > 1$.
\end{proof}

\section{Distributed Relaxed Greedy Algorithm}
\label{sec:algorithm}

We now describe a distributed version of the relaxed greedy
algorithm from Section~\ref{sec:relaxedGreedyAlgorithm}. Like the
sequential relaxed greedy algorithm, this algorithm also runs in
$O(\log n)$ phases  --- with edges in  $E_i$ being processed in
phase $i$. We will show that edges in $E_0$ can be processed in
$O(1)$ rounds. Recall that each subsequent phase consists of the
following five steps: (i) computing a  cluster cover of $G'_{i-1}$,
(ii) selecting query edges in $E_i$, (iii) computing a cluster graph
$H_{i-1}$ of $G'_{i-1}$, (iv) answering shortest path queries for
selected query edges, and (v) deleting some redundant edges. We will
show that Steps (ii), (iii), and (iv) can be completed in $O(1)$
rounds and Steps (i) and (v) take $O(\log^* n)$ rounds. Step (i) and
Step (v) will each involve computing an MIS in a certain derived
graph and in both cases, we will show that the derived graph is a
UBG that resides in a metric space of constant doubling dimension.
Putting this all together, we will show that the algorithm runs in
$O(\log n \cdot \log^* n)$ communication rounds.

\subsection{Distributed Processing of Short Edges}
Lemma~\ref{lem:clique} implies that vertices in the same component
of $G_0 = G[E_0]$ induce a clique and therefore can communicate in
one hop with each other. In the distributed version of the
algorithm, each vertex $u$ obtains the topology of its closed
neighborhood along with pairwise distances between neighbors in one
hop. Using this information, $u$ determines the connected component
$C$ of $G_0$ that it belongs to. Then $u$ simply runs {\tt
SEQ-GREEDY} on $C$ and computes a $t$-spanner of $C$. Finally, $u$
identifies the edges of the $t$-spanner incident on itself and
informs all its neighbors of this.
%Note that since {\tt
%SEQ-GREEDY} is deterministic, every node in a connected
%component $C$ computes the same $t$-spanner. This redundant
%computation might be avoided by picking a leader in each connected
%component and assigning it the task of running {\tt
%SEQ-GREEDY}.

\begin{theorem}
The edges in $E_0$ can be processed in $O(1)$ rounds of communication.
\end{theorem}

\subsection{Distributed Processing of Long Edges}
\label{subsection:longerEdges.distr}
In this section, we show how long edges, that is, edges in $E_i$, $i
> 0$, can be processed in a distributed setting. The first step of
this process is the computation of a cluster cover for the spanner
$G'_{i-1}$ updated at the end of the previous phase.

\subsubsection{Distributed Cluster Cover for $G'_{i-1}$}
Recall that in this step our goal is to compute a cluster cover
$\{C_{u_1}, C_{u_2}, \ldots\}$ of $G'_{i-1}$ of radius $\delta
W_{i-1}$. To do this, each node $u$ first identifies all nodes $v$
in $G$ satisfying $\ssp_{G'_{i-1}}(u, v) \le \delta W_{i-1}$. Using
arguments similar to those in Section \ref{sec:query.answer}, we can
show that any node $v$ satisfying $\ssp_{G'_{i-1}}(u, v) \le \delta
W_{i-1}$ must be at most $2\delta W_{i-1}/\alpha$ hops from $u$.
%\begin{lemma}
%If $sp_{G_{i-1}}(u, v) \le \delta W_{i-1}$, then $sp_{G}(u, v)$
%contains at most $2 \delta W_{i-1}/ \alpha$ hops.
%\label{lem:constant.hops}
%\end{lemma}
%\begin{proof}
%Clearly $sp_{G}(u, v) \le sp_{G_{i-1}}(u, v)$. Suppose now that
%$sp_{G_{i-1}}(u, v) \le \delta W_{i-1}$ and let $sp_{G}(u, v) = x_0,
%x_1, \ldots x_k$, with $x_0 = u$ and $x_k = v$. By the triangle
%inequality we have that $d(x_i, x_{i+2}) < d(x_{i}, x_{i+1})+
%d(x_{i+1}, x_{i+2})$, for each $i = 0, 1, \ldots k-2$. This means
%that $\{x_i, x_{i+2}\}$ is not an edge in $G$, otherwise the path
%obtained by replacing $\{x_i, x_{i+1}\}$ and $\{x_{i+1}, x_{i+2}\}$
%with $\{x_i, x_{i+2}\}$ would be shorter than $sp_{G}(u, v)$, a
%contradiction. So we have that
%\begin{tabbing}
%%................\=..................................\=.....\=................\=...\=........\=...\kill
%\>$d(x_0, x_1) + d(x_1, x_2)$ \> $>$ \> $d(x_0, x_2)$ \> $>$
%$\alpha$ \\
%\>$d(x_2, x_3) + d(x_3, x_4)$ \> $>$ \> $d(x_2, x_4)$ \> $>$
%$\alpha$ \\
%\> $\cdots$
%\end{tabbing}
%Summing up these inequalities gives us $sp_{G}(u, v) > \alpha k/2$
%and using the upper bound $\delta W_{i-1}$ on $sp_G(u, v)$, we get
%that $k < 2 \delta W_{i-1}/\alpha \le 2 \delta / \alpha$.
%\end{proof}
So each node $u$ constructs the subgraph of $G'_{i-1}$ induced by
nodes that are at most $2 \delta W_{i-1}/\alpha$ hops away from it
in $G$.
%While obtaining this local view of $G'_{i-1}$, node $u$ also
%obtains pairwise distances between nodes in this local view.
Node
$u$ then runs a (sequential) single source shortest path algorithm
with source $u$ on the local view of $G'_{i-1}$ it has obtained and
identifies all nodes $v$ satisfying $\ssp_{G'_i}(u, v)
\le \delta W_{i-1}$.
%Node $u$ then sends a message telling $v$ that
%it belongs to cluster $C_u$. This entire process takes a constant
%number of communication rounds.

At the end of the above process, every node $u$ in the network is
a cluster center.
%the head of a cluster $C_u$ of radius $\delta W_{i-1}$.
%Also, every
%node in the network belongs to every cluster centered at a node at
%distance at most $\delta W_{i-1}$ away (in $G'_{i-1}$).
We now force some nodes to cease being cluster centers, so that all
pairs of cluster centers are far enough from each other. Let $J$ be
the graph with vertex set $V$ and whose edges $\{x, y\}$ are such
that $x \in C_y$ (and by symmetry, $y \in C_x$).

\begin{lemma}
$J$ is a UBG that resides in a metric space of constant doubling
dimension. \label{lem:j.growth.bounded}
\end{lemma}
\begin{proof}
For any edge $\{x, y\}$ in $J$, we have that $x \in C_y$ and
therefore $\ssp_{G'_{i-1}}(x, y) \le \delta W_{i-1}$. Assign to
every pair of nodes $\{x, y\}$ in $V$ a weight $w(x, y) =
\ssp_{G'_{i-1}}(x, y)$. The weights $w$ form a metric simply because
shortest path distances in any graph form a metric. Thus $J$ is a
graph whose nodes reside in a metric space and whose edges connect
pairs of nodes separated by distance of at most $\delta W_{i-1}$ (in
the metric space). By scaling the quantity $\delta W_{i-1}$ up to
one, we see that $J$ is a UBG in the underlying metric space defined
by the weights $w$. Recall from \cite{KuhnMoscibrodaWattenhofer}
that the {\em doubling dimension} of a metric space is the smallest
$\rho$ such that every ball can be covered by at most $2^\rho$ balls
of half the radius. To see that the metric space induced by the
weights $w$ has constant doubling dimension, start with a ball of
$B$ radius $R$ centered at an arbitrary vertex $u$. Every vertex $v$
in ball $B$ satisfies $\ssp_J(u, v) \le R$. Now cover the vertices
in $B$ using balls of radius $R/2$ as follows: repeatedly pick an
uncovered vertex $v$ in $B$ and grow a ball of radius $R/2$ centered
at $v$, until all vertices have been covered. We now show that the
number of balls of radius $R/2$ is constant.

Let $a$ and $b$ be two arbitrary centers of different balls of
radius $R/2$. Then $\ssp_J(u, v) > R/2$, otherwise $a$ and $b$ would
belong to the same ball of radius $R/2$. We distinguish three
situations:
\begin{itemize}
\item $\{a, b\}$ is not an edge in $G$. This implies that $|ab| >
\alpha$.
\item $\{a, b\}$ is an edge in $G$ that has not been processed prior
to phase $i$. This implies that $|ab| > W_{i-1} \ge 1/\delta$ (after
scaling $\delta W_{i-1}$ up to $1$).
\item $\{a, b\}$ is an edge in $G$ that has been processed prior
to phase $i$. This implies that $G'_{i-1}$ contains a $t$-spanner
path from $a$ to $b$ and therefore $|ab| \ge \ssp_{G'_{i-1}}(a, b) /
t \ge R/2t$.
\end{itemize}
We have established that $|ab| \ge \min\{\alpha, \frac{1}{\delta},
\frac{R}{2t}\}$, so no two ball centers can be too close to each
other. It follows that the number of balls of radius $R/2$ that fit
inside $B$ is constant, proving the lemma true.
\end{proof}

Let $I$ be an MIS of $J$ constructed using the MIS algorithm in
\cite{KuhnMoscibrodaWattenhofer}.
This algorithm runs in $O(\log^* n)$ communication rounds on a UBG
that resides in a metric space of constant doubling dimension.
Then each node in $V \setminus I$ has one or more neighbors
in $I$. Each node $u \in I$ is declared a cluster center, and each
node $v \in V \setminus I$ attaches itself to the neighbor in $I$
with the highest identifier.
This gives us the desired cluster cover of radius $\delta W_{i-1}$.

\begin{theorem}
A cluster cover of $G'_{i-1}$ of radius $\delta W_{i-1}$ can be
computed in $O(log^* n)$ rounds of communication.
\end{theorem}
%\begin{proof}
%Each node $u$ identifies nodes $v$ that might belong to $C_u$ in
%constant number of rounds, as discussed above. The algorithm from
%\cite{KuhnMoscibrodaNiebergWattenhofer,KuhnMoscibrodaWattenhofer}
%takes poly-logarithmic number of communication rounds on
%growth-bounded graphs. Since $J$ is growth-bounded
%(Lemma~\ref{lem:j.growth.bounded}), computing a Maximal Independent
%Set $I$ of $J$ takes $O(\log^{*} n)$ communication rounds.
%\end{proof}

% We haven't used this lemma anywhere.
%\begin{lemma}
%For any two cluster centers $u$ and $v$, $d(u, v) \ge \frac{\delta
%W_{i-1}}{t}$. \label{lem:cluster.centers}
%\end{lemma}
%\begin{proof}
%First note that $sp_{G_{i-1}}(u, v) > \delta W_{i-1}$; otherwise,
%$u$ and $v$ would be neighbors in $J$, contradicting the fact that
%$u$ and $v$ are independent in $I$. If $d(u, v) > W_{i-1}$, then
%clearly $d(u, v) \ge \frac{\delta W_{i-1}}{t}$, since $0 < \delta <
%1/2$ and $t > 1$. So suppose that $d(u, v) \le W_{i-1}$. In this
%case we have that $sp_{G_{i-1}}(u, v) \le t \cdot d(u, v)$, since
%$G_{i-1}$ contains a $t$-spanner path for any pair of vertices no
%more than $W_{i-1}$ distance apart. This along with the fact that
%$sp_{G_{i-1}}(u, v) > \delta W_{i-1}$ completes the proof.
%\end{proof}

%\smallskip\noindent
%Lemma~\ref{lem:clustercenter.far} tells us that cluster centers are
%reasonably far apart; very short edges of $G_{i-1}$, which may
%induce a high cost in answering shortest path queries, are hidden
%inside clusters.
\smallskip
\subsubsection{Distributed Query Edge Selection}
\label{sec:query.pairs.distr}

Only nodes that are cluster heads need to participate in the process
of selecting query edges. Each cluster head $a$ seeks to gather
information on all edges in $E_i$ between the cluster $C_a$ and any
other cluster $C_b$. Using the argument in Section
\ref{sec:query.answer}, we know that every node in $C_a$ is at most
$2\delta W_{i-1}/\alpha$ hops away from $a$ in $G$. Therefore, if
there is an edge $\{u, v\} \in E_i$, $u \in C_a$ and $v \in C_b$,
then $v$ is at most $1 + 2\delta W_{i-1}/\alpha$ hops away from $a$.
So $a$ gets information from nodes that are at most $1 + 2\delta
W_{i-1}/\alpha$ hops away from it and
%also knowing the cluster cover restricted to this subgraph,
it identifies all edges in $E_i[C_a,
C_b]$. Recall that this is the set of edges in $E_i$ which connect a
node in $C_a$ and a node in $C_b$. Node $a$ then discards all
covered edges from $E_i[C_a, C_b]$, leaving only candidate query
edges in $E_i$ between $C_a$ and $C_b$. Finally, from among the
candidate query edges, node $a$ selects an edge $\{u, v\}$ that
minimizes
$t \cdot |uv| - \ssp_{G'_{i-1}}(a,u) - \ssp_{G'_{i-1}}(b, v)$.
%\label{eq:query.pairs.u}
%Note that after node $a$ has constructed the edge sets $E_i[C_a,
%C_b]$, all other computations that it performs are completely local.
%Node $a$ then passes on its selection of query edges to all
%neighbors.
%from~(\ref{eq:query.pairs.u}) among all neighbors in $C_a$.
%\end{enumerate}
%} \fbox{\box0}
%\hrule\hfill
%\end{minipage}\end{center}
\begin{theorem}
Query edges from $E_i$ can be selected in $O(1)$ rounds of communication.
\end{theorem}

\subsubsection{Distributed Construction of the Cluster Graph}
\label{sec:distr.cluster.graph}

As in the query edge selection step, only the cluster heads need to
perform actions to compute the cluster graph. Any member $u$ of a
cluster $C_a$ lies at most $2\delta W_{i-1}/\alpha$ hops away from
$a$ in $G$. Thus $a$ can identify intra-cluster edges incident on it
by gathering information from at most $2\delta W_{i-1}/\alpha$ hops
away. If $C_b$ is a cluster with $\ssp_{G'_{i-1}}(a, b) \le
W_{i-1}$, then node $a$ can identify the inter-cluster edge $\{a,
b\}$ by gathering information from at most $2W_{i-1}/\alpha$ hops
away. If $C_b$ is a cluster such that there is an edge $\{u, v\}$ in
$G'_{i-1}$ with $u \in C_a$ and $v \in C_b$, then node $a$ can
identify the inter-cluster edge $\{a, b\}$ by gathering information
from at most $2(2\delta + 1)W_{i-1}/\alpha$ hops away. Note that the
information that $a$ gathers contains a local view of $G'_{i-1}$
along with all pairwise distances. Using this information, node $a$
is able to run a single source shortest path algorithm with source
$a$ and determine the weights of all inter-cluster and intra-cluster
edges incident on $a$.

\begin{theorem}
Computing the cluster graph $H_{i-1}$ of $G'_{i-1}$ takes
$O(1)$ communication rounds.
\end{theorem}

\subsubsection{Answering Shortest Path Queries}

Each node $u$ knows all the query edges incident on it. As proved in
Section \ref{sec:query.answer}, node $u$ only needs to gather
information from nodes that are at most a constant number of hops
away, to be able to determine locally, for all incident query edges
$\{u, v\} \in E_i$, whether $\ssp_{H_{i-1}}(u, v) \le t \cdot |uv|$.
Thus, after constant number of communication rounds, $u$ knows the
subset of incident query edges $\{u, v\}$ for which
$\ssp_{H_{i-1}}(u, v) > t \cdot |uv|$ and $u$ identifies these as
the incident edges to be added to $G'_i$.

\begin{theorem}
Answering shortest path queries takes $O(1)$ communication rounds.
\end{theorem}

\subsubsection{Distributed Removal of Redundant Edges}
Two edges $\{u, v\}$ and $\{u', v'\}$ in $G'_i$ are mutually
redundant if (i) $\ssp_{H_{i-1}}(u, u') + |u'v'| +
\ssp_{H_{i-1}}(v', v) \le t_1 \cdot |uv|$ and
 (ii) $\ssp_{H_{i-1}}(u', u) + |uv| +
\ssp_{H_{i-1}}(v, v') \le t_1 \cdot |u'v'|$.
%, and (iii) $|u
%v'| \le \alpha$ and $|vu'| \le \alpha$ (here we assume that $u, v,
%u', v'$ is a clockwise order of these nodes).
%Figure~\ref{fig:redundant}).
Each node $u$ takes charge of all edges $\{u, v\}$ added to $G_i$ in
phase $i$ and for which the identifier of $u$ is higher than the
identifier of $v$. For each such edge $\{u, v\}$ that $u$ is in
charge of, $u$ determines all edges $\{u', v'\}$ such that $\{u,
v\}$ and $\{u', v'\}$ form a mutually redundant pair. Note that the
nodes $u$ and $v'$ are a constant number of hops away from each
other in $G$, and similarly for nodes $v$ and $u'$.
%above implies that $\{u,
%v'\}$ and $\{v, u'\}$ are both in $G$, meaning that $u$ can
%determine all such edges $\{u', v'\}$ in $O(1)$ communication
%rounds. i
Node $u$ then contributes to the construction of the graph $J$ by
adding to $V(J)$ a vertex for each redundant edge $u$ is in charge
of, and to $E(J)$ an edge connecting nodes in $V(J)$ that correspond
to mutually redundant edges in $G_i$.
%Using an argument similar to
%the one used in Lemma~\ref{lem:j.growth.bounded}, we can show the
%following: % property of $J$:
We now show the following property of $J$:

\begin{lemma}
$J$ is a UBG that resides in a metric space of constant doubling
dimension. \label{lem:j.growth.bounded2}
\end{lemma}
\begin{proof}
Let $a$ and $b$ be vertices in $J$ corresponding to edges $\{u_a,
v_a\}$ and $\{u_b, v_b\}$ in $G'_i$. Assign to the vertex pair $(a,
b)$ a weight equal to
\begin{eqnarray}
\nonumber d_J(a, b) & =  & \min(\ssp_{H_{i-1}}(u_a,u_b) +
\ssp_{H_{i-1}}(v_a, v_b), \\
\nonumber   &   &            ~~~\ssp_{H_{i-1}}(u_a,v_b) +
\ssp_{H_{i-1}}(v_a, u_b)).
\end{eqnarray}
First we show that the weights defined by $d_J$ form a metric.
Clearly $d_J(a,a) = 0$ and $d_J(a, b) = d_J(b, a)$. To prove the
triangle inequality, consider three vertices $a, b, c \in J$. Assume
w.l.o.g. that
\begin{tabbing}
.......\=...............\=.......\=............\=...\kill
  \> \(d_J(a,b)\) \> \(=\) \> \(\ssp_{H_{i-1}}(u_a,u_b) + \ssp_{H_{i-1}}(v_a,v_b)\) \\
  \> \(d_J(b,c)\) \> \(=\) \> \(\ssp_{H_{i-1}}(u_b,u_c) + \ssp_{H_{i-1}}(v_b, v_c)\)
\end{tabbing}
\begin{figure}[htpb]
\centerline{
\includegraphics[width = 0.9\linewidth]{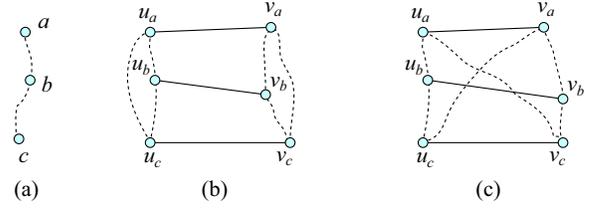}}
\caption{$d_J$ is a metric: (a) Nodes $a, b, c \in J$ correspond to
edges $\{u_a, u_b\}$, $\{u_b, u_c\}$, $\{u_a, u_c\} \in G'_i$. (b)
$d_J(a, c) = \ssp_{H_{i-1}}(u_a,u_c) + \ssp_{H_{i-1}}(v_a,v_c)$ (c)
$d_J(a, c) = \ssp_{H_{i-1}}(u_a,v_c) + \ssp_{H_{i-1}}(v_a,u_c)$.}
\label{fig:triangle}
%\vspace{-2em}
\end{figure}
We identify two possible scenarios:
\begin{enumerate}
\item [(1)] $d_J(a, c) = \ssp_{H_{i-1}}(u_a,u_c) + \ssp_{H_{i-1}}(v_a,v_c)$ (see Figure~\ref{fig:triangle}b).
Since $\ssp$ is itself a metric, it follows immediately that $d_J(a,
c) \le d_J(a,b) + d_J(b, c)$.
\item [(2)] $d_J(a, c) = \ssp_{H_{i-1}}(u_a,v_c) + \ssp_{H_{i-1}}(v_a,u_c)$ (see Figure~\ref{fig:triangle}c).
Then it must be that $d_J(a,c) \le \ssp_{H_{i-1}}(u_a,u_c)$ $+$
$\ssp_{H_{i-1}}(v_a, v_c) \le d_J(a,b) + d_J(b, c)$, cf. scenario
(1).
\end{enumerate}
We have shown that $d_J$ defines a metric. We now show that $J$ is a
quasi-UBG residing in the metric space defined by $d_J$. For each
edge $\{a, b\}$ in $J$, the following redundancy conditions hold:
\begin{itemize}
\item [(a)] $\ssp_{H_{i-1}}(u_a, u_b) + \ssp_{H_{i-1}}(v_b, v_a) \le
t_1 \cdot |u_a v_a| - |u_b v_b|$
\item [(b)] $\ssp_{H_{i-1}}(u_b, u_a) + \ssp_{H_{i-1}}(v_a, v_b) \le t_1 \cdot
|u_b v_b| - |u_a v_a|$
\end{itemize}
Recall that $\{u_a, v_a\}$ and $\{u_b, v_b\}$ are both in $E_i$, for
some $i \ge 0$. This implies that their lengths differ by a factor
of $r$ at the most: $W_{i-1} < |u_a v_a| \le r \cdot W_{i-1}$ and
$W_{i-1} < |u_b v_b| \le r \cdot W_{i-1}$. Thus the right hand side
of inequalities (a) and (b) above is a quantity that lies in the
interval $((t_1-r)W_{i-1}, (t_1r-1)W_{i-1})$. By scaling
$(t_1r-1)W_{i-1}$ up to one we can say that $J$ is an
$\frac{t_1-r}{t_1r-1}$ - qUBG in the underlying metric space defined
by $d_J$.

It remains to show that the metric space defined by $d_J$ has
constant doubling dimension. Throughout the rest of the proof we use
$B_J$ ($B_H$) to denote a ball in the metric space defined by $d_J$
($\ssp_{H_{i-1}}$).

\begin{figure}[htpb]
\centerline{
\includegraphics[width = 0.9\linewidth]{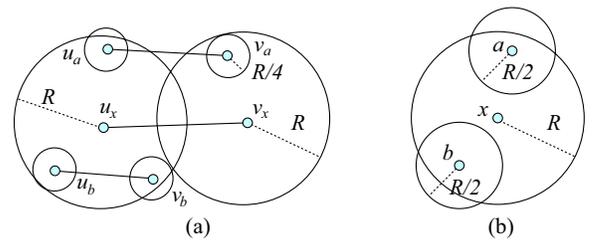}}
\caption{(a) The metric space defined by (a) $\ssp_{H_{i-1}}$ and
(b) $d_J$, has doubling dimension.} \label{fig:metric}
%\vspace{-2em}
\end{figure}

Consider a ball $B_J(x, R)$ of radius $R$ centered at an arbitrary
vertex $x \in J$ corresponding to edge $\{u_x, v_x\} \in H_{i-1}$.
Let $a, b \in J$ be such that $d_J(a, b) > R/2$ (see
Figure~\ref{fig:metric}). Assume w.l.o.g that $d_J(a, b)$ $= $
$\ssp_{H_{i-1}}(u_a, u_b)$ $+$ $\ssp_{H_{i-1}}(v_a, v_b)$. Then at
least one of the following must be true:
\begin{itemize}
\item [(i)] $\ssp_{H_{i-1}}(u_a, u_b) > R/4$.
\item[(ii)] $\ssp_{H_{i-1}}(v_a, v_b) > R/4$.
\end{itemize}
We use these observations, along with the fact that $\ssp_{H_{i-1}}$
defines a metric space of constant doubling dimension, to show that
$d_J$ defines a metric space of constant doubling dimension.

%That the doubling dimension of $d_{G'_{i-1}}$ is constant follows
%immediately from Lemma~\ref{lem:virtual.spanner.leap} and the fact
%that the doubling dimension of $d$ is constant.

%To cover all relevant vertices in $B_H(u_a, R) \cup B_H(v_a, R)$, do
%the following repeatedly: (i) pick an uncovered vertex $u_b$ of an
%edge $\{u_b, v_b\}$ corresponding to $b \in J$, (ii) grow two balls
%$B_H(u_b, R/2)$ and $B_H(v_b, R/2)$, which eventually cover new
%vertices, and (iii) grow the corresponding ball $B_J(b, R/2)$.

To cover all vertices in $B_J(x, R)$, do the following repeatedly:
(i) pick an uncovered vertex $a \in J$ (ii) grow a ball $B_J(a,
R/2)$, and (iii) grow two balls $B_H(u_a, R/4)$ and $B_H(v_a, R/4)$
in $H_{i-1}$, where $a = \{u_a, v_a\}$.

Arguments similar to the ones used in
Lemma~\ref{lem:j.growth.bounded} show that, for any ball centers
$u_a, u_b \in B_H(u_x, R)$, we have that $|u_au_b| \ge \min\{\alpha,
\frac{1}{\delta}, \frac{R}{2t}\}$. Since no two ball centers can be
too close to each other, it follows that $B_H(u_x, R)$ gets covered
by a constant number of balls of radius $R/4$, and similarly for
$B_H(v_x, R)$. Conform observation (ii) above, corresponding to each
uncovered $a \in J$, there is an uncovered vertex $u_a$ or $v_a$ in
$H_{i-1}$. These together show that the number of balls covering
$B_J(x, R)$ is constant, thus completing the proof.
\end{proof}

\smallskip
\noindent
Let $I$ be an MIS of $J$ constructed using the MIS
algorithm in \cite{KuhnMoscibrodaWattenhofer} that takes $O(\log^*
n)$ communication rounds on a UBG that resides in a metric space of
constant doubling dimension. Each node $u$ then removes from $G_i$
all incident edges in $V(J) \setminus I$.

\begin{theorem}
Removing redundant edges takes $O(\log^* n)$ communication rounds.
\end{theorem}

\section{Future work}
The results presented in this paper apply to $\alpha$-UDGs embedded
in constant-dimension Euclidean spaces, and do not directly
generalize to doubling metric spaces. For low dimensional doubling
metric spaces, we believe it possible to construct an $O(\log n
\log^* n)$ distributed algorithm that produces a
$(1+\varepsilon)$-spanner with constant maximum degree. However, new
techniques may be needed for lightweight spanners; the techniques
presented in this paper use a key property (the leapfrog property)
that does not seem to generalize to metrics of doubling dimension.

%\bibliographystyle{plain}
%\bibliography{distSpanner}

\def\cprime{$'$}

\end{document}